\documentclass[journal,a4paper]{IEEEtran}
\usepackage[T1]{fontenc}
\usepackage[utf8]{inputenc}
\usepackage{bm}
\usepackage{amsmath}
\usepackage{mathtools}
\usepackage{amssymb}
\usepackage{amsthm}
\usepackage{graphicx}
\usepackage{algorithm}
\usepackage{algorithmic}
\usepackage{cite} 
\usepackage[bookmarks=false,
 breaklinks=false,pdfborder={0 0 0},colorlinks=false]
 {hyperref}

\newcommand{\E}{\mathbb{E}}
\newcommand{\C}{\mathbb{C}}
\DeclareMathOperator{\diag}{diag}
\DeclareMathOperator{\tr}{tr}
\DeclareMathOperator{\rank}{rank}
\DeclareMathOperator{\blkdiag}{blkdiag}
\DeclareMathOperator*{\argmin}{arg\,min}

\usepackage[caption=false,font=footnotesize]{subfig}

\newtheorem{proposition}{Proposition}

\theoremstyle{definition}

\theoremstyle{remark}
\newtheorem{remark}{Remark}

\begin{document}
\title{Exact Payload-Decoupling Conditions for Pilot-Only BEM Channel Estimation With Application to OTFS}
\author{Gianmarco~Romano, Francesco~A.~N.~Palmieri,~\IEEEmembership{Member,~IEEE,}
Stefano~Buzzi, \IEEEmembership{Fellow,~IEEE},
Giovanni~Di~Gennaro, and~Amedeo~Buonanno, \IEEEmembership{Senior Member,~IEEE}
\thanks{Part of this work was presented at the IEEE International Workshop on Signal Processing Advances in Wireless Communications (SPAWC 2026)~\cite{romano2026}.}
\thanks{G. Romano, F.A.N. Palmieri and G. Di Gennaro are with Dipartimento di Ingegneria, Università degli Studi della Campania ``Luigi~Vanvitelli'', Aversa (CE), Italy (e-mails:~\{gianmarco.romano; francesco.palmieri; giovanni.digennaro\}@unicampania.it).}
\thanks{S. Buzzi is with University of Cassino and Southern Lazio, Cassino, Italy, and with Consorzio Nazionale Interuniversitario per le Telecomunicazioni, Parma, Italy (e-mail: buzzi@unicas.it).}
\thanks{A. Buonanno is with Dept. of Energy Technologies and Renewable Sources, ENEA, Portici (NA), Italy (e-mail: amedeo.buonanno@enea.it).}
}
\markboth{}{Romano \MakeLowercase{\emph{et al.}}: Exact Payload-Decoupling Conditions for Pilot-Only Basis-Expansion-Model Channel Estimation}
\maketitle

\begingroup
\footnotesize
\centering
\emph{This work has been submitted to the IEEE for possible publication. Copyright may be transferred without notice, after which this version may no longer be accessible.}\par
\endgroup
\vspace{0.6em}

\begin{abstract}
In high-mobility doubly dispersive links, basis expansion models (BEMs) reduce channel dimensionality, yet unknown payload symbols generally contaminate conventional matched-pilot channel estimates. This paper establishes the exact conditions under which such estimates become payload-independent and derives a pilot, guard, and data-placement rule that guarantees these conditions. We prove a necessary-and-sufficient zero pilot--data interference (ZPDI) condition under which the matched-pilot least-squares (LS) solution coincides with the maximum-likelihood (ML) estimator for the reduced pilot statistic. When ZPDI holds, estimation requires a single precomputed projection. When it does not, the estimate contains a deterministic, channel-scaled payload bias that persists at high signal-to-noise ratio. A disjoint-support rule, independent of the selected basis, realizes ZPDI through pilot, guard, and data placement. We then specialize the framework to orthogonal time--frequency space (OTFS) and examine its structural and performance consequences. With the generalized complex-exponential BEM (GCE-BEM), the ZPDI estimator remains within about $2$~dB of the perfect channel state information benchmark in bit error rate at speeds up to 500~km/h.
\end{abstract}

 
\IEEEpeerreviewmaketitle{}

\section{Introduction}

\IEEEPARstart{A}{ccurate} channel state information (CSI) supports coherent detection, adaptive modulation, and resource allocation, and is therefore required for reliable wireless-system operation~\cite{tse2005,goldsmith2005}. High-mobility scenarios, including high-speed railway (HSR) links, unmanned aerial vehicle (UAV) communications, and low-Earth-orbit (LEO) satellite systems, involve rapidly time-varying channels with large Doppler spreads and severe inter-carrier interference (ICI)~\cite{chen2023,khuwaja2018,darya2025}. In these doubly dispersive environments, conventional orthogonal frequency-division multiplexing (OFDM) loses subcarrier orthogonality~\cite{hlawatsch2011}. Orthogonal time--frequency space (OTFS) modulation offers an alternative because high-mobility channels appear approximately sparse and quasi-stationary in the delay--Doppler (DD) domain~\cite{hadani2017,wei2021,hong2022}. Practical OTFS channels, however, often have fractional Doppler and fractional delay, which spread energy over several DD bins and reduce this sparsity~\cite{hong2022,deng2025}.

Existing channel estimators operate either as standalone methods or jointly with symbol detection. Standalone methods provide a channel estimate before payload detection, whereas joint receivers couple or alternate the two tasks. The available approaches include delay--Doppler domain techniques~\cite{raviteja2019,yuan2021}, time--frequency domain methods~\cite{sheng2024}, and Bayesian, matching-pursuit, message-passing, and artificial-intelligence-based frameworks~\cite{aslandogan2025,li2025}. These approaches still face practical limitations. Methods that accommodate fractional parameters often incur high computational cost, whereas deep learning (DL) solutions require substantial training data, computation, and energy; these requirements can be prohibitive when propagation conditions change dynamically~\cite{li2025}.

Basis expansion models (BEMs) provide a structured parametric representation of time-varying channels for pilot-assisted estimation. They project the channel response onto a finite-dimensional subspace, reducing the number of unknown parameters while retaining the relevant channel dynamics~\cite{giannakis1998,ma2003a,hlawatsch2011}. BEM-based channel estimation has been studied for several modulation schemes and, more recently, for OTFS and waveforms such as affine frequency division multiplexing (AFDM). In practice, channel state information must be obtained primarily from pilot observations and may subsequently be refined through data-aided processing. Even in a data-aided receiver, the first stage must estimate the channel before reliable payload decisions are available. Pilot--data interference (PDI), however, degrades this initial estimate and has prompted several mitigation methods. Recent studies embed BEMs in joint channel-estimation and detection frameworks~\cite{huang2023,guo2025,guo2024}. Huang \textit{et al.}~\cite{huang2023}, for example, exploit decoded data symbols in a turbo-equalization loop, whereas Guo \textit{et al.}~\cite{guo2025} use generalized compressed sensing and iterative cancellation to mitigate residual PDI. Qu \textit{et al.}~\cite{qu2021a} instead adopt a discrete prolate spheroidal sequence (DPSS) subspace with linear minimum mean-squared error (LMMSE) estimation and treat PDI statistically. Liu \textit{et al.}~\cite{liu2022} derive a generalized complex-exponential BEM (GCE-BEM) OTFS input--output model, but their receiver still depends on iterative data-aided equalization.

Closed-form, non-iterative BEM estimators are also available for specific OTFS settings. Mohebbi \textit{et al.}~\cite{mohebbi2026} derive a two-dimensional discrete prolate spheroidal BEM (DPS-BEM) representation for multiple-input multiple-output OTFS (MIMO-OTFS) and recover the BEM coefficients by least squares under a particular pilot structure. Wu \textit{et al.}~\cite{wu2025} develop a GCE-BEM estimator for OTFS with hardware impairments using a single-pilot/guard construction. Their statistic nevertheless contains pilot- and data-dependent terms, and the data-dependent component is treated as effective interference or noise. Thus, despite their closed forms, these methods remain tied to particular BEM choices, pilot structures, statistical assumptions, or treatments of interference.

Our central result is a necessary-and-sufficient characterization of when the conventional matched-pilot BEM estimator is payload-independent and coincides with the maximum-likelihood (ML) estimator for the reduced statistic. Under the resulting zero pilot--data interference (ZPDI) condition, the statistic contains no data term for any admissible payload and coefficient vector. If ZPDI fails, the same filter contains a deterministic, channel-scaled payload bias that persists at high signal-to-noise ratio (SNR) for any fixed pilot-to-data power ratio. We therefore identify the exact boundary between the validity of the least-squares (LS) solution as a precomputable pilot-only ML estimator and its structural contamination by unknown data. Beyond this boundary, the matched-pilot architecture must tolerate, model, or cancel the interference through data-aided processing.

A pilot--data orthogonality condition of this kind has appeared before. For cyclic-prefixed affine transmissions over doubly selective channels described by the complex-exponential BEM (CE-BEM), Kannu and Schniter~\cite{kannu2008} showed that pilot--data orthogonality, together with an optimal-excitation condition, is necessary and sufficient for the Bayesian LMMSE channel estimate to attain its mean-squared error (MSE) lower bound. The present formulation asks a different question, covers a broader model class, and keeps payload decoupling separate from excitation. ZPDI characterizes when the conventional matched-pilot projection yields an exactly payload-decoupled, prior-free estimate within the selected BEM subspace, rather than when a statistically optimal estimator exists. Because the condition is stated at the operator level, it covers arbitrary BEM subspaces, any unitary observation transform, and the cyclic-prefix, zero-padding, and chirp-periodic-prefix guard structures, rather than only the CE-BEM and time-domain cyclic-prefix setting of~\cite{kannu2008}. In our formulation, the excitation condition also has a distinct and more general role: a scaled-identity pilot Gram matrix governs the estimator conditioning and the associated noise enhancement, but not whether the estimate is payload-decoupled. Separating these requirements leaves the Gram matrix as a design parameter. It need not be a scaled identity for the estimate to remain exactly payload-independent and can be deliberately non-orthogonal to trade conditioning for Doppler resolution, as in the oversampled GCE-BEM.

These distinctions have direct consequences for implementation and error analysis. Unlike iterative data-aided receivers~\cite{liu2022,huang2023} and statistics-dependent designs~\cite{qu2021a,wu2025}, the ZPDI estimator needs neither payload decisions nor stochastic channel priors. Its matrices depend only on the pilot pattern, basis, and system dimensions; they can therefore be precomputed, leaving an online stage that is linear in the frame length. Once PDI is removed exactly, the remaining error separates into thermal noise shaped by the pilot Gram matrix and modeling error controlled by the BEM. The framework replaces heuristic or simulation-only acceptance of a pilot arrangement with a rigorous design criterion and shows whether decoupling, conditioning, or BEM resolution limits the estimate.

The paper develops the payload-decoupling characterization, a design rule that realizes it, and the resulting structural and computational consequences:

\begin{enumerate}
    \item We prove necessary and sufficient conditions under which the matched pilot projection is exactly decoupled from unknown payload symbols, for arbitrary BEM subspaces and block-linear modulation structures that admit the operator model considered here. The result determines when the conventional matched-pilot LS solution is a valid payload-independent ML estimator for the reduced observation, rather than a data-contaminated estimate.
    \item We provide a disjoint-support sufficient condition that is independent of the chosen BEM basis, turning the ZPDI requirement into a concrete rule for placing pilots, guards, and data so that ZPDI holds across arbitrary BEM representations and the considered block-linear modulation structures.
    \item We characterize the pilot Gram matrix in terms of identifiability, conditioning, and structure. The analysis shows how the number of pilot observations limits the admissible BEM order, how Gram-matrix conditioning controls noise enhancement, and how zero inter-tap pilot leakage turns the offline inversion from a full dense solve into independent per-tap solves, with a further scalar simplification when the per-tap blocks are diagonal.
    \item We develop one detailed OTFS application using the standard full-guard pilot~\cite{raviteja2019}. We prove simultaneous ZPDI and zero pilot--pilot leakage, derive the Toeplitz GCE-BEM Gram matrix and its resolution--conditioning trade-off, and illustrate these results at speeds up to 500~km/h.
\end{enumerate}

The remainder of the paper is organized as follows. Section~\ref{sec:channel-model} presents the system model and BEM representation. Section~\ref{sec:estimation-problem} develops the pilot-projected LS/ML channel estimator and derives the zero pilot--data interference condition. Section~\ref{sec:gram-properties} analyzes the pilot Gram matrix and its role in identifiability, conditioning, and complexity. Section~\ref{sec:otfs-examples} applies the framework to OTFS modulation and analyzes the resulting GCE-BEM structure. Finally, Section~\ref{sec:numerical_examples} presents numerical results, and Section~\ref{sec:conclusion} concludes the paper.

\textit{Notation:} Bold uppercase ($\mathbf{A}$) and lowercase ($\mathbf{a}$) letters denote matrices and vectors, respectively. Superscripts $(\cdot)^{T}$, $(\cdot)^{H}$, and $(\cdot)^{*}$ denote transpose, conjugate transpose, and complex conjugate. $\mathbf{A}^{-1}$ denotes the inverse. Operators $\diag(\cdot)$, $\blkdiag(\cdot)$, and $\operatorname{vec}(\cdot)$ denote diagonal matrices, block-diagonal matrices, and column-wise vectorization. $\|\cdot\|$ and $\|\cdot\|_{F}$ denote Euclidean and Frobenius norms, while $\odot$ and $\otimes$ denote Hadamard and Kronecker products. $\mathbf{1}[\cdot]$ denotes the indicator function, and $\E\{\cdot\}$ denotes statistical expectation.

\section{System Model and BEM Representation}
\label{sec:channel-model}

We consider the transmission of a block represented by vector $\mathbf{x} \in \C^{N}$ through a doubly dispersive channel characterized by its discrete-time impulse response $h[n,\ell]$, where the first index $n$ denotes \emph{time} (sample index) and the second index $\ell$ denotes \emph{delay} (tap index). The time dependence of $h[n,\ell]$ captures the Doppler-induced variation of each multipath component across the block, while the finite range $\ell\in\{0,\ldots,L-1\}$ reflects bounded delay spread. The receiver observes additive white Gaussian noise (AWGN) represented by vector $\mathbf{v}\in\C^{N}$ with zero mean and covariance $\sigma_v^2\mathbf{I}$.

To eliminate intersymbol interference arising from multipath propagation, a guard of $N_g\geq L-1$ samples is added to each transmitted block. After removing the guard at the receiver, the $N$-sample useful portion of the received signal satisfies
\begin{equation}
    y[n] = \sum_{\ell=0}^{L-1} h[n, \ell]\, x[n - \ell] + v[n],
    \label{eq:scalar_convolution_useful}
\end{equation}
for $n = 0, 1, \ldots, N-1$. This equation describes the time-varying convolution between the transmitted signal and the channel~\cite{tse2005}. With the tap-based finite-impulse-response (FIR) parametrization, the estimation model does not require the physical Doppler shifts to be integer multiples of the bin resolution. In contrast, delay--Doppler parametrizations must explicitly model the inter-bin leakage caused by off-grid shifts~\cite{benzine2026}.

The input--output relationship \eqref{eq:scalar_convolution_useful} can be written as the matrix--vector product
\begin{equation}
    \mathbf{y} = \mathbf{H} \mathbf{x} + \mathbf{v},
    \label{eq:matrix_io_relation}
\end{equation}
where the effective channel matrix $\mathbf{H}\in\C^{N\times N}$ is constructed as follows. The channel coefficient matrix $\mathbf{G}\in\C^{N\times L}$, with columns $\mathbf{g}_\ell=[h[0,\ell],\ldots,h[N-1,\ell]]^T\in\C^N$ for $\ell=0,\ldots,L-1$, collects the time-varying tap sequences over the block. After guard insertion and removal, the effective interference-free channel matrix decomposes as
\begin{equation}
\mathbf{H}=\sum_{\ell=0}^{L-1}\diag(\mathbf{g}_{\ell})\,\mathbf{S}_{\ell},
\label{eq:Hguard}
\end{equation}
where $\mathbf{S}_\ell\in\C^{N\times N}$ is the guard-dependent shift operator. The three guard modes considered here give the following shift structures:
\begin{itemize}
  \item Cyclic Prefix (CP): $\mathbf{S}_{\ell,\mathrm{cp}}=\boldsymbol{\Pi}^{\ell}$, where $\boldsymbol{\Pi}$ is the $N\times N$ cyclic permutation matrix
  \item Zero Padding (ZP): $[\mathbf{S}_{\ell,\mathrm{zp}}]_{n,k}=\mathbf{1}[n=k+\ell]$, the nilpotent lower-shift of order $\ell$, which treats wrapped samples as zero rather than recycling them.
  \item Chirp Periodic Prefix (CPP): $\mathbf{S}_{\ell,\mathrm{cpp}}=\boldsymbol{\Omega}_\ell\boldsymbol{\Pi}^\ell$, where $\boldsymbol{\Omega}_\ell=\diag\!\left(\omega_\ell[k]\right)_{k=0}^{N-1}$ with $\omega_\ell[k]=e^{j2\pi c_1(\ell^2-2\ell k)/N}$ for $k=0,\ldots,\ell-1$ and $\omega_\ell[k]=1$ otherwise, so that the chirp phase correction applies only to the $\ell$ cyclically wrapped samples~\cite{bemani2023,benzine2026}.
\end{itemize}

The decomposition \eqref{eq:Hguard} reveals that the effective channel matrix is a superposition of $L$ time-varying scaled-shift terms, one per delay tap, with the shift structure fully determined by the guard mode.

The transmit vector $\mathbf{x}$ carries both pilot and data symbols. A symbol vector $\mathbf{s}\in\C^{N_s}$ collects $N_p$ known pilot symbols $\mathbf{s}_p\in\C^{N_p}$ and $N_u$ unknown payload symbols $\mathbf{s}_u\in\C^{N_u}$, with $N_s=N_p+N_u$. Binary selection matrices $\mathbf{P}_p\in\{0,1\}^{N_s\times N_p}$ and $\mathbf{P}_u\in\{0,1\}^{N_s\times N_u}$ place each class into the common vector as
\begin{equation}
\mathbf{s}
=
\mathbf{P}_p\mathbf{s}_p
+
\mathbf{P}_u\mathbf{s}_u,
\qquad
\mathbf{P}_p^{T}\mathbf{P}_u=\mathbf{0},
\label{eq:symbol_decomposition}
\end{equation}
where the constraint $\mathbf{P}_p^{T}\mathbf{P}_u=\mathbf{0}$ enforces disjoint pilot and payload positions before modulation.

A coding matrix $\mathbf{C}\in\C^{N\times N_s}$ then maps the symbol vector into the transmitted block $\mathbf{x}=\mathbf{C}\mathbf{s}$. At the receiver, a unitary
transformation $\mathbf{D}\in\C^{N\times N}$ ($\mathbf{D}^H\mathbf{D}=\mathbf{I}$)
maps the received block into the desired signal domain. The unified input-output relation is
\begin{equation}
\mathbf{r}
=
\mathbf{D}\!\left(\sum_{\ell=0}^{L-1}\diag(\mathbf{g}_{\ell})\mathbf{S}_\ell\right)
\mathbf{C}\left(\mathbf{P}_p\mathbf{s}_p+\mathbf{P}_u\mathbf{s}_u\right)
+\mathbf{w},
\label{eq:transform-domain-system-model}
\end{equation}
where $\mathbf{w}=\mathbf{D}\mathbf{v}\sim\mathcal{CN}(\mathbf{0},\sigma_w^2\mathbf{I})$.

Estimating the full $NL$-element time-varying channel matrix $\mathbf{G}$ directly from pilot observations is a high-dimensional problem. For realistic mobile channels satisfying the underspread property, the true channel degrees of freedom are much smaller than the apparent dimensionality $NL$ of $\mathbf{G}$.

BEMs exploit this structural redundancy by representing the channel with $Q$ basis functions and $QL$ coefficients, where $Q \ll N$. The effective channel dimension is thereby reduced from $NL$ to $QL$ parameters, at the cost of a controlled modeling error and with a substantial reduction in pilot overhead and estimation complexity. The general BEM representation of the channel coefficient matrix is
\begin{equation}
\mathbf{G} = \boldsymbol{\Phi} \boldsymbol{\Gamma} + \mathbf{E},
\label{eq:BEM}
\end{equation}
where $\boldsymbol{\Phi}\in\C^{N\times Q}$ is the basis matrix, $\boldsymbol{\Gamma}\in\C^{Q\times L}$ is the coefficient matrix, and $\mathbf{E}\in\C^{N\times L}$ is the modeling error containing the channel components outside the BEM subspace. The choice of $Q$ must satisfy two competing requirements. It should be much smaller than $N$ to provide meaningful dimensionality reduction, but large enough to keep the modeling error acceptably small under the anticipated channel conditions.

Substitution of the BEM representation \eqref{eq:BEM} into the system model \eqref{eq:transform-domain-system-model} directly relates the received signal to the BEM coefficients. We first define the guard-shifted transmit vectors
\begin{equation}
    \mathbf{z}_{\ell} = \mathbf{S}_{\ell}\mathbf{x}, \quad \ell = 0, 1, \ldots, L-1,
\label{eq:shifted-symbol}
\end{equation}
where each vector is the transmit block $\mathbf{x}$ shifted by $\mathbf{S}_\ell$ according to the guard structure of the $\ell$-th delay tap. We arrange these vectors in the composite block matrix
\begin{equation}
    \mathbf{Z}(\mathbf{x}) = [\diag(\mathbf{z}_{0}), \ldots, \diag(\mathbf{z}_{L-1})]\in\C^{N\times NL},
\end{equation}
which places the transmit contribution of all delay taps in a form compatible with the time-varying channel response. Define the BEM embedding matrix as
\begin{equation}
\mathbf{B}\triangleq\mathbf{I}_{L}\otimes\boldsymbol{\Phi}\in\C^{NL\times QL},
\label{eq:bem-embedding}
\end{equation}
and using $\operatorname{vec}(\mathbf{G})=\mathbf{B}\boldsymbol{\gamma} + \operatorname{vec}(\mathbf{E})$ with $\boldsymbol{\gamma}=\operatorname{vec}(\boldsymbol{\Gamma})\in\C^{QL}$, the BEM sensing matrix
\begin{equation}
\boldsymbol{\Psi}(\mathbf{x})=\mathbf{D}\mathbf{Z}(\mathbf{x})\mathbf{B}\in\C^{N\times QL}
\label{eq:Psi-def}
\end{equation}
contains the joint effects of the transmit signal, the basis functions, and the observation-domain transformation. With the effective modeling error $\boldsymbol{\eta}(\mathbf{x})=\mathbf{D}\mathbf{Z}(\mathbf{x})\operatorname{vec}(\mathbf{E})$, the system model becomes
\begin{equation}
    \mathbf{r}=
    \boldsymbol{\Psi}_{p}\boldsymbol{\gamma} +
    \boldsymbol{\Psi}_{u}(\mathbf{s}_u)\boldsymbol{\gamma} +
    \boldsymbol{\eta}(\mathbf{x}) +
    \mathbf{w},
\label{eq:system-model-compact-bem}
\end{equation}
where $\boldsymbol{\Psi}(\mathbf{x})=\boldsymbol{\Psi}_{p}+\boldsymbol{\Psi}_{u}(\mathbf{s}_u)$ follows from the linearity of $\mathbf{Z}(\cdot)$. This separation also identifies which terms are known at the receiver. The pilot contribution $\boldsymbol{\Psi}_{p}\boldsymbol{\gamma}$ supports pilot-only channel estimation, whereas the payload contribution $\boldsymbol{\Psi}_{u}(\mathbf{s}_u)\boldsymbol{\gamma}$ depends on unknown data and can interfere with the estimate. The receiver can construct the pilot sensing matrix before detection, but it cannot construct the payload sensing matrix until the data symbols are known. The modeling error $\boldsymbol{\eta}$ measures the mismatch between the true channel and its BEM approximation, and its effect depends on the basis functions and on $Q$.

\section{Pilot-Projected Channel Estimation and Zero Pilot--Data Interference}
\label{sec:estimation-problem}

The BEM representation reduces channel estimation to the recovery of $\boldsymbol{\gamma}$ from the observation \eqref{eq:system-model-compact-bem}. The receiver must estimate these coefficients from the known pilots before equalization, but the system model contains error terms that the receiver cannot reconstruct.

In practice, the receiver adopts the observation model
\begin{equation}
\mathbf{r}
=
\boldsymbol{\Psi}_{p}\boldsymbol{\gamma}
+
\boldsymbol{\Psi}_{u}(\mathbf{s}_u)\boldsymbol{\gamma}
+
\mathbf{w},
\label{eq:assumed-pilot-payload-model}
\end{equation}
where the modeling error \(\boldsymbol{\eta}(\mathbf{x})\) is omitted. Any out-of-subspace channel component, however, remains as a mismatch that does not average out with the noise and therefore bounds the achievable accuracy at high SNR.

The assumed model still depends on the unknown data symbols, and the receiver must in general solve a joint detection problem. The joint ML estimator of the payload and the channel coefficients minimizes the coupled LS cost
\begin{equation}
\begin{aligned}
(\hat{\mathbf{s}}_u,\hat{\boldsymbol{\gamma}})
&=
\argmin_{\substack{\mathbf{s}_u\in\mathcal{A}^{N_u}\\
\boldsymbol{\gamma}\in\C^{QL}}}\;
\mathcal{E}(\mathbf{s}_u,\boldsymbol{\gamma}),\\
\mathcal{E}(\mathbf{s}_u,\boldsymbol{\gamma})
&\triangleq
\left\|
\mathbf{r}
-
\boldsymbol{\Psi}(\mathbf{s}_u)\boldsymbol{\gamma}
\right\|^2 .
\end{aligned}
\label{eq:ls_objective}
\end{equation}

Consider first the genie-aided case in which the receiver knows the transmitted payload block \(\mathbf{s}_u\). For a fixed \(\mathbf{s}_u\), the sensing matrix \(\boldsymbol{\Psi}(\mathbf{s}_u)\) is known and \eqref{eq:assumed-pilot-payload-model} becomes a standard linear Gaussian model in \(\boldsymbol{\gamma}\). If \(\boldsymbol{\Psi}(\mathbf{s}_u)\) has full column rank, the LS estimate is
\begin{equation}
\hat{\boldsymbol{\gamma}}
=
\left(
\boldsymbol{\Psi}^{H}(\mathbf{s}_u)
\boldsymbol{\Psi}(\mathbf{s}_u)
\right)^{-1}
\boldsymbol{\Psi}^{H}(\mathbf{s}_u)\mathbf{r},
\label{eq:ls_solution_known_symbols}
\end{equation}
which coincides with the ML estimate under the Gaussian assumption. Since the modeling error has been dropped from \eqref{eq:assumed-pilot-payload-model}, even the ML solution may not recover the true $\boldsymbol{\gamma}$.

Decision-directed receivers approximate the target in \eqref{eq:ls_solution_known_symbols}. Iterative data-aided schemes, such as the BEM-OTFS receiver of Liu \textit{et al.}~\cite{liu2022}, begin with a pilot-based channel estimate, use it to detect the payload \(\hat{\mathbf{s}}_u\), rebuild the sensing matrix \(\boldsymbol{\Psi}(\hat{\mathbf{s}}_u)\), and refine \(\hat{\boldsymbol{\gamma}}\). This refinement exploits the payload, but its accuracy depends on the quality of the initial estimate because reliable payload decisions are unavailable before the channel is known. Errors in the initial estimate can therefore propagate through the payload decisions and into subsequent channel updates. Equation~\eqref{eq:ls_solution_known_symbols} cannot be used in the first stage, so a standalone pilot-only estimator is needed to initialize such a receiver.

We consider the conventional matched-pilot estimator, which uses the known pilot sensing matrix \(\boldsymbol{\Psi}_{p}\) because the payload contribution in \eqref{eq:assumed-pilot-payload-model} is not available at this stage. Applying the matched pilot filter \(\boldsymbol{\Psi}_{p}^{H}\) projects the received vector onto the pilot subspace and produces the \(QL\)-dimensional pilot statistic
\begin{equation}
\begin{aligned}
\mathbf{t}_p
\triangleq
\boldsymbol{\Psi}_p^H\mathbf{r}
&=
\mathbf{R}_{pp}\boldsymbol{\gamma}
+
\mathbf{R}_{pu}(\mathbf{s}_u)\boldsymbol{\gamma}
+
\widetilde{\mathbf{w}}_p ,
\end{aligned}
\label{eq:projected_observation_with_interference}
\end{equation}
where
\begin{equation}
\mathbf{R}_{pp}
\triangleq
\boldsymbol{\Psi}_p^H\boldsymbol{\Psi}_p,
\qquad
\mathbf{R}_{pu}(\mathbf{s}_u)
\triangleq
\boldsymbol{\Psi}_p^H\boldsymbol{\Psi}_u(\mathbf{s}_u),
\label{eq:Rpp-Rpu-defs}
\end{equation}
and the projection colors the thermal noise as
\begin{equation}
\widetilde{\mathbf{w}}_p
\triangleq
\boldsymbol{\Psi}_p^H\mathbf{w}
\sim
\mathcal{CN}
\left(
\mathbf{0},
\sigma_w^2\mathbf{R}_{pp}
\right).
\label{eq:matched-pilot-noise}
\end{equation}
The matched filter is the natural reduction for estimating \(\boldsymbol{\gamma}\) from the pilots, but it does not separate the pilot and payload contributions. The pilot statistic still carries the PDI term
\begin{equation}
\boldsymbol{\zeta}_p(\mathbf{s}_u)
\triangleq
\mathbf{R}_{pu}(\mathbf{s}_u)\boldsymbol{\gamma},
\label{eq:equivalent-disturbance}
\end{equation}
which arises from the cross-correlation \(\mathbf{R}_{pu}(\mathbf{s}_u)\) between the pilot and payload sensing matrices. Because it lies in the same \(QL\)-dimensional subspace used to recover \(\boldsymbol{\gamma}\), it cannot be separated from the useful term \(\mathbf{R}_{pp}\boldsymbol{\gamma}\) by linear filtering of \(\mathbf{t}_p\) alone.

The severity of the PDI becomes explicit when the ordinary matched-pilot LS filter \(\mathbf{R}_{pp}^{-1}\) is applied to the pilot statistic, which gives
\begin{equation}
\mathbf{R}_{pp}^{-1}\mathbf{t}_p
=
\boldsymbol{\gamma}
+
\mathbf{R}_{pp}^{-1}
\mathbf{R}_{pu}(\mathbf{s}_u)\boldsymbol{\gamma}
+
\mathbf{R}_{pp}^{-1}\widetilde{\mathbf{w}}_p .
\label{eq:pdi-contaminated-pilot-estimate}
\end{equation}
The first term is the desired channel, and the third is a noise contribution that vanishes as \(\sigma_w^2\rightarrow 0\). The second term is the PDI, which, unlike the noise, does not vanish at high SNR. Raising the pilot power relative to the payload attenuates but does not cancel it, and averaging over the noise leaves it unchanged, because for a given transmitted block the PDI is a deterministic function of the payload actually sent. A standalone matched-pilot estimator therefore cannot recover the channel reliably unless the PDI term is eliminated or rendered negligible.

The proposed approach enforces this requirement directly. If the pilot and payload sensing matrices are designed so that \(\mathbf{R}_{pu}(\mathbf{s}_u)=\mathbf{0}\) for every admissible payload vector, the PDI term disappears from \eqref{eq:projected_observation_with_interference} for every channel realization. The resulting estimator is independent of both the transmitted payload realization and any assumed payload distribution in this model. The known pilot Gram matrix then yields a closed-form matched-pilot LS/ML estimator for the reduced pilot statistic.
    
\subsection{Zero Pilot--Data Interference Condition}
\label{sec:zero-pilot-interference}

The pilot statistic in \eqref{eq:projected_observation_with_interference} depends on the unknown payload through the term \(\mathbf{R}_{pu}(\mathbf{s}_u)\boldsymbol{\gamma}\). The matched-pilot estimator is therefore decoupled from the data symbols only if
\begin{equation}
\mathbf{R}_{pu}(\mathbf{s}_u)
=
\boldsymbol{\Psi}_p^H\boldsymbol{\Psi}_u(\mathbf{s}_u)
= \mathbf{0},\quad\forall\,\mathbf{s}_u,
\label{eq:orthogonality-condition}
\end{equation}
which we term the ZPDI condition.

\begin{proposition}[Matched-Pilot Decoupling]
\label{prop:optimal-estimator}
Assume that \(\boldsymbol{\Psi}_{p}\) has full column rank, and consider the matched-pilot estimator \(\hat{\boldsymbol{\gamma}}=\mathbf{R}_{pp}^{-1}\mathbf{t}_p\). Within the selected BEM subspace, this estimator recovers \(\boldsymbol{\gamma}\) exactly in the noiseless model for every admissible payload vector and every coefficient vector if and only if the ZPDI condition \eqref{eq:orthogonality-condition} holds. Under \eqref{eq:orthogonality-condition},
\begin{equation}
\hat{\boldsymbol{\gamma}}
=
\mathbf{R}_{pp}^{-1}\mathbf{t}_p
=
\left(\boldsymbol{\Psi}_{p}^{H}\boldsymbol{\Psi}_{p}\right)^{-1}\boldsymbol{\Psi}_{p}^{H}\mathbf{r},
\label{eq:optimal_solution_zero_interference}
\end{equation}
which depends only on the pilot matrix $\boldsymbol{\Psi}_{p}$ and the received signal $\mathbf{r}$, and coincides with the ML estimator for the reduced pilot-projected observation \(\mathbf{t}_p\) in \eqref{eq:projected_observation_with_interference} under Gaussian noise with covariance \eqref{eq:matched-pilot-noise}.
\end{proposition}

\begin{IEEEproof}
In the noiseless BEM model, the estimation error is, by \eqref{eq:projected_observation_with_interference}, \(\hat{\boldsymbol{\gamma}}-\boldsymbol{\gamma}=\mathbf{R}_{pp}^{-1}\mathbf{R}_{pu}(\mathbf{s}_u)\boldsymbol{\gamma}\). If \eqref{eq:orthogonality-condition} holds, this error vanishes for every \(\mathbf{s}_u\) and \(\boldsymbol{\gamma}\), which proves sufficiency. Conversely, exact recovery for every \(\boldsymbol{\gamma}\) and every admissible \(\mathbf{s}_u\) requires \(\mathbf{R}_{pp}^{-1}\mathbf{R}_{pu}(\mathbf{s}_u)\boldsymbol{\gamma}=\mathbf{0}\) for all \(\boldsymbol{\gamma}\) and \(\mathbf{s}_u\); nonsingularity of \(\mathbf{R}_{pp}\) then gives \(\mathbf{R}_{pu}(\mathbf{s}_u)=\mathbf{0}\) for every admissible payload vector, which proves necessity. Finally, under ZPDI, \(\mathbf{t}_p\sim\mathcal{CN}(\mathbf{R}_{pp}\boldsymbol{\gamma},\sigma_w^2\mathbf{R}_{pp})\), and maximization of the Gaussian likelihood over \(\boldsymbol{\gamma}\) yields \(\hat{\boldsymbol{\gamma}}=\mathbf{R}_{pp}^{-1}\mathbf{t}_p\).
\end{IEEEproof}

Proposition~\ref{prop:optimal-estimator} characterizes the validity domain of the familiar matched-pilot formula, not new LS algebra. Its necessity result applies specifically to the matched pilot statistic and does not exclude more general oblique projections of the complete observation. Such processing lies outside the reduced statistic \(\mathbf{t}_p\) and may enhance noise. ZPDI instead makes the standard pilot projection payload-independent without this correction.

Hereafter, the term \emph{ZPDI estimator} denotes the standard LS/ML
estimator in \eqref{eq:optimal_solution_zero_interference} when the pilot
pattern satisfies the ZPDI condition.

\begin{remark}[Relation to the Minimum Mean-Squared-Error Pilot-Aided Transmission Conditions of~\cite{kannu2008}]
\label{rem:kannu-relation}
Setting $\mathbf{D}=\mathbf{I}$, $\mathbf{S}_\ell=\boldsymbol{\Pi}^{\ell}$, and
$\boldsymbol{\Phi}$ equal to the CE-BEM basis reduces
\eqref{eq:orthogonality-condition} to the pilot--data orthogonality condition
of~\cite{kannu2008}, while the optimal-excitation condition there corresponds to
$\mathbf{R}_{pp}=(\mathcal{E}_p/N)\mathbf{I}$, i.e., the
$\kappa(\mathbf{R}_{pp})=1$ matched-filter case of
Section~\ref{subsec:complexity}. Here ZPDI and the Gram condition play separate
roles: ZPDI alone decides payload independence, for arbitrary BEM subspaces and the
guard structures of~\eqref{eq:Hguard}, while the Gram matrix is treated separately as
an identifiability, conditioning, and complexity parameter
(Section~\ref{sec:gram-properties}). This separation admits estimators whose Gram
matrix is deliberately not a scaled identity, such as the oversampled GCE-BEM of
Section~\ref{sec:otfs-examples}, which therefore remain valid while lying outside the
optimal-excitation characterization of~\cite{kannu2008}.
\end{remark}

Condition \eqref{eq:orthogonality-condition} admits a direct geometric reading. The columns of $\boldsymbol{\Psi}_p$ are the channel-shifted pilot waveforms through which the matched filter $\boldsymbol{\Psi}_p^H$ observes the BEM coefficient vector $\boldsymbol{\gamma}$, and those of $\boldsymbol{\Psi}_u(\mathbf{s}_u)$ are the corresponding payload waveforms. The requirement $\boldsymbol{\Psi}_p^H\boldsymbol{\Psi}_u(\mathbf{s}_u)=\mathbf{0}$ places the payload column space in the orthogonal complement of the pilot column space for every admissible $\mathbf{s}_u$, so the payload falls in a subspace the estimator does not observe and the term $\mathbf{R}_{pu}(\mathbf{s}_u)\boldsymbol{\gamma}$ in \eqref{eq:projected_observation_with_interference} vanishes for any channel realization.

The pilot-design problem thus becomes geometric: decoupling is obtained by shaping the supports of pilots, guards, and data within the frame so that their channel-shifted images occupy orthogonal subspaces.

This decoupling is deterministic. Unlike statistical interference suppression, decision-directed refinement, or iterative joint estimation and detection \eqref{eq:ls_objective}, whose estimates depend on the payload, its statistics, or tentative decisions, the ZPDI estimator reduces to the single matrix--vector product \eqref{eq:optimal_solution_zero_interference}; its Gram matrix $\mathbf{R}_{pp}$ depends only on the pilot pattern, the BEM basis, and the system dimensions, and is precomputed once and reused for every frame. The estimate is then limited not by pilot--data interference but only by thermal noise, the conditioning of $\mathbf{R}_{pp}$, the finite BEM modeling error, and residual support or synchronization mismatch.

The orthogonality $\boldsymbol{\Psi}_p^H\boldsymbol{\Psi}_u(\mathbf{s}_u)=\mathbf{0}$ is a single global condition. For pilot design, however, it is more informative to examine the condition one delay tap at a time. Writing
\(\boldsymbol{\Psi}_p=[\boldsymbol{\Psi}_{p,0},\ldots,\boldsymbol{\Psi}_{p,L-1}]\)
and similarly for the payload, ZPDI is guaranteed when
\begin{equation}
\boldsymbol{\Psi}_{p,\ell}^{H}\boldsymbol{\Psi}_{u,m}=\mathbf{0}\quad\forall\,\ell,m\in\{0,1,\ldots,L-1\},
\label{eq:orthogonality_requirement}
\end{equation}
where $\boldsymbol{\Psi}_{p,\ell} = \mathbf{D}\diag(\mathbf{z}_{p,\ell})\boldsymbol{\Phi}$ and $\boldsymbol{\Psi}_{u,m} = \mathbf{D}\diag(\mathbf{z}_{u,m})\boldsymbol{\Phi}$. Each block $\boldsymbol{\Psi}_{p,\ell}^{H}\boldsymbol{\Psi}_{u,m}$ measures the leakage of the data carried at delay $m$ into the pilot observation at delay $\ell$. Since the channel can place energy on any of the $L$ taps, a single index pair is not enough: complete decoupling requires all $L^{2}$ blocks to vanish simultaneously, so that the pilots remain orthogonal to the data after every admissible delay shift.

Condition \eqref{eq:orthogonality_requirement} still depends on the BEM basis $\boldsymbol{\Phi}$ and can therefore be checked only after the basis has been fixed. A stronger and simpler rule removes this dependence by acting directly on the time-domain pilot and data vectors before application of the basis. This rule requires more than disjoint pilot and data positions in the transmitted block. For \emph{every} pair of delay shifts $(\ell,m)\in\{0,\ldots,L-1\}^{2}$ applied independently to the pilot and data through the shift operators of \eqref{eq:Hguard}, the shifted pilot and shifted data must occupy disjoint positions. Disjoint unshifted supports alone do not ensure this property: a delay shift can move the pilot onto a data position, or the data onto a pilot position, even when their original supports are disjoint.

\begin{proposition}[Disjoint Support Condition for Zero Interference]
\label{prop:disjoint-support}
Consider the block-structured observation matrices $\boldsymbol{\Psi}_{p,\ell} = \mathbf{D}\diag(\mathbf{z}_{p,\ell})\boldsymbol{\Phi}$ and $\boldsymbol{\Psi}_{u,m} = \mathbf{D}\diag(\mathbf{z}_{u,m})\boldsymbol{\Phi}$, where $\mathbf{D}^H\mathbf{D} = \mathbf{I}$. A sufficient condition for $\boldsymbol{\Psi}_{p,\ell}^{H}\boldsymbol{\Psi}_{u,m}=\mathbf{0}$ for all $\ell,m$ is
\begin{equation}
\mathbf{z}_{p,\ell}^{*}\odot\mathbf{z}_{u,m}=\mathbf{0}\quad\forall\,\ell,m\in\{0,1,\ldots,L-1\}.
\label{eq:disjoint_support_condition}
\end{equation}
\end{proposition}

\begin{IEEEproof}
The cross-correlation block can be expressed as
\begin{align}
\boldsymbol{\Psi}_{p,\ell}^{H}\boldsymbol{\Psi}_{u,m} & =\boldsymbol{\Phi}^{H}\diag(\mathbf{z}_{p,\ell}^{*})\mathbf{D}^{H}\mathbf{D}\diag(\mathbf{z}_{u,m})\boldsymbol{\Phi}\nonumber\\
 & =\boldsymbol{\Phi}^{H}\diag(\mathbf{z}_{p,\ell}^{*}\odot\mathbf{z}_{u,m})\boldsymbol{\Phi},
\label{eq:hadamard_product_form}
\end{align}
where the second equality uses $\mathbf{D}^{H}\mathbf{D}=\mathbf{I}$ and $\diag(\mathbf{a})\diag(\mathbf{b})=\diag(\mathbf{a}\odot\mathbf{b})$. If \eqref{eq:disjoint_support_condition} holds, the middle diagonal matrix vanishes and the result follows for all $\ell,m$.
\end{IEEEproof}

The disjoint-support condition \eqref{eq:disjoint_support_condition} is a practical design rule. The vectors $\mathbf{z}_{p,\ell}$ and $\mathbf{z}_{u,m}$ are the pilot and data seen through delay shifts $\ell,m\in\{0,\ldots,L-1\}$, and the condition asks that no shifted pilot land on a shifted data position, so pilot and data must be separated by the full delay support on either side. This is precisely the role of the guard: a narrower guard lets a shifted pilot overlap the data and reintroduces $\mathbf{R}_{pu}(\mathbf{s}_u)\boldsymbol{\gamma}$, whereas a wider one only consumes resource elements without improving the decoupling.

\begin{remark}[Use as an iterative-receiver initializer]
The estimator in \eqref{eq:optimal_solution_zero_interference} is available before payload detection and can therefore initialize decision-directed BEM receivers. After detection, one may rebuild $\boldsymbol{\Psi}(\hat{\mathbf{s}}_u)$ and apply the data-aided LS update in \eqref{eq:ls_solution_known_symbols}. Thus the ZPDI estimator can be used either as a stand-alone pilot-only estimator or as the first stage of an iterative receiver, such as~\cite{liu2022,qu2021a}, without modification.
\end{remark}

\section{Identifiability, Conditioning, and Complexity of the ZPDI Estimator}
\label{sec:gram-properties}

Once $\mathbf{R}_{pu}(\mathbf{s}_u)=\mathbf{0}$, the payload term disappears from the pilot statistic and the estimator is governed entirely by the pilot Gram matrix $\mathbf{R}_{pp}$. This matrix describes the inner geometry of the pilot observation, rather than the level of interference: it is the Gram matrix of the channel-shifted pilot waveforms through which $\boldsymbol{\gamma}$ is observed. The guard must annihilate $\mathbf{R}_{pu}$, whereas, after removal of the data term, $\mathbf{R}_{pp}$ determines whether $\boldsymbol{\gamma}$ can be recovered and how sensitive that recovery is. Its rank sets coefficient \emph{identifiability}; its spectrum sets estimator \emph{conditioning} and the associated noise enhancement; and its structure sets the \emph{computational cost} of the inversion. The following subsections examine these roles separately.

\subsection{Identifiability}
\label{sec:identifiability} 

The ZPDI condition removes the payload from the pilot statistic, but it does not guarantee that the $QL$ BEM coefficients in $\boldsymbol{\gamma}$ can be recovered. The estimator \eqref{eq:optimal_solution_zero_interference} is well defined only if $\boldsymbol{\Psi}_p$ has full column rank, that is, if
\begin{equation} 
\rank(\boldsymbol{\Psi}_p)=QL,
\qquad\text{equivalently}\qquad
\mathbf{R}_{pp}\succ\mathbf{0}.
\label{eq:identifiability-condition}
\end{equation}

This requirement is what limits the BEM order. The pilot observation matrix $\boldsymbol{\Psi}_p=\mathbf{D}\mathbf{Z}_p\mathbf{B}$ has $QL$ columns, but it is built from the known pilot excitation $\mathbf{Z}_p\triangleq\mathbf{Z}(\mathbf{x}_p)$, which produces only $N_{\mathrm{obs}}\triangleq\rank(\mathbf{Z}_p)$ independent scalar observations. The basis $\mathbf{B}$ cannot create observations that the pilots do not excite, so $\rank(\boldsymbol{\Psi}_p)\le N_{\mathrm{obs}}$ and full column rank requires
\begin{equation}
N_{\mathrm{obs}}\ge QL .
\label{eq:counting-condition}
\end{equation}
The admissible BEM order is therefore capped at
\begin{equation}
Q \le \left\lfloor \frac{N_{\mathrm{obs}}}{L}\right\rfloor .
\label{eq:q-identifiability-ceiling}
\end{equation}
The bound is a direct count: the $L$ delay taps carry $Q$ coefficients each, for $QL$ unknowns, while the pilots supply only $N_{\mathrm{obs}}$ independent equations. Choosing $Q$ above \eqref{eq:q-identifiability-ceiling} makes $\boldsymbol{\Psi}_p$ rank deficient, so distinct coefficient vectors produce the same pilot statistic and no unique estimate exists. The count is necessary but not sufficient, since the retained basis functions may still be linearly dependent under the pilot operator; the operative requirement remains \eqref{eq:identifiability-condition}.

\subsection{Conditioning and Noise Enhancement}
\label{sec:conditioning}

Identifiability guarantees that $\mathbf{R}_{pp}$ is invertible, but not that the inversion is well conditioned. Under ZPDI the pilot statistic reduces to $\mathbf{t}_p=\mathbf{R}_{pp}\boldsymbol{\gamma}+\widetilde{\mathbf{w}}_p$, with $\widetilde{\mathbf{w}}_p$ given by \eqref{eq:matched-pilot-noise}, so once the payload is removed the estimate applies $\mathbf{R}_{pp}^{-1}$ to a noisy pilot statistic and its accuracy depends on how this inverse acts on the matched noise. When $\mathbf{R}_{pp}$ is close to a scaled identity, all coefficient directions are treated comparably and the matched noise is not amplified. When $\mathbf{R}_{pp}$ is nearly singular, the identifiability condition \eqref{eq:identifiability-condition} still holds, so the solution remains unique, but small perturbations of the pilot statistic translate into large errors in $\boldsymbol{\gamma}$. The condition number $\kappa(\mathbf{R}_{pp})$ measures this sensitivity and sets the worst-case noise enhancement of the ZPDI estimator.

\subsection{Zero Pilot--Pilot Leakage and Inversion Complexity}
\label{subsec:complexity}

By analogy with ZPDI, we define the \emph{zero pilot--pilot leakage} (ZPPL) condition as
\begin{equation}
\boldsymbol{\Psi}_{p,i}^{H}\boldsymbol{\Psi}_{p,j}
=\mathbf{0},
\quad \forall\,i\neq j.
\label{eq:zero_pilot_leakage_condition}
\end{equation}
This condition nulls the cross-Gram matrices between the pilot observations of different delay taps. It therefore governs whether the coefficient groups of distinct taps remain coupled once ZPDI is established.

Partition $\mathbf{R}_{pp}$ into $L\times L$ blocks $\mathbf{R}_{ij}\triangleq\boldsymbol{\Psi}_{p,i}^{H}\boldsymbol{\Psi}_{p,j}$, so that $\mathbf{R}_{ij}$ is the cross-Gram matrix of taps $i$ and $j$. Under ZPPL every off-diagonal block vanishes, hence
\begin{equation}
\mathbf{R}_{pp}
=\blkdiag\!\left(\mathbf{R}_{00},\ldots,
\mathbf{R}_{L-1,L-1}\right),
\label{eq:block_diagonal_normal_equations}
\end{equation}
and, provided that each $Q\times Q$ diagonal block is nonsingular, the estimator separates into $L$ per-tap estimators,
\begin{equation}
\hat{\boldsymbol{\gamma}}
=\begin{bmatrix}
\mathbf{R}_{00}^{-1}\boldsymbol{\Psi}_{p,0}^{H}\\
\vdots\\
\mathbf{R}_{L-1,L-1}^{-1}\boldsymbol{\Psi}_{p,L-1}^{H}
\end{bmatrix}\mathbf{r}.
\label{eq:block_diagonal_estimation}
\end{equation}
Thus ZPPL makes $\mathbf{R}_{pp}$ block diagonal, but not necessarily diagonal. Diagonality requires, in addition, that distinct BEM functions within the same delay tap be uncorrelated. Since
\begin{equation}
[\mathbf{R}_{\ell\ell}]_{a,b}
=\sum_{n=0}^{N-1}|z_{p,\ell}[n]|^{2}
\phi_{a}^{*}[n]\phi_{b}[n],
\label{eq:within_tap_pilot_gram}
\end{equation}
$\mathbf{R}_{pp}$ is diagonal if and only if ZPPL holds and
\begin{equation}
\sum_{n=0}^{N-1}|z_{p,\ell}[n]|^{2}
\phi_{a}^{*}[n]\phi_{b}[n]=0,
\quad \forall\,\ell,\ a\neq b.
\label{eq:within_tap_orthogonality_condition}
\end{equation}
Condition~\eqref{eq:within_tap_orthogonality_condition} is a pilot-weighted orthogonality condition on the BEM functions. The diagonal entries may still differ, so diagonality alone does not imply ideal conditioning.

These structural properties translate directly into the inversion cost. Let $d=QL$ be the number of unknowns. A dense $\mathbf{R}_{pp}$ requires an $O(d^{3})=O(Q^{3}L^{3})$ inversion, whereas under ZPPL the $L$ independent blocks reduce this cost to $O(LQ^{3})$ and can be inverted in parallel. When \eqref{eq:within_tap_orthogonality_condition} also holds, $\mathbf{R}_{pp}$ is diagonal and the inversion reduces to $QL$ scalar reciprocals.

Since the pilot pattern, the BEM basis, and the system dimensions are fixed, the projection
\begin{equation}
\mathbf{W}_{p}=\mathbf{R}_{pp}^{-1}\boldsymbol{\Psi}_{p}^{H}
\label{eq:precomputed-filter}
\end{equation}
can be precomputed once and stored. Each frame then requires only the product $\hat{\boldsymbol{\gamma}}=\mathbf{W}_{p}\mathbf{r}$, with $O(NQL)$ online complexity. The benefit of ZPPL is therefore confined to the offline inversion, which it makes cheaper and parallelizable, while the online complexity order remains the same unless further pilot sparsity is exploited.

 \section{OTFS Channel Estimation and GCE-BEM Analysis}
\label{sec:otfs-examples}

We apply the preceding framework to OTFS modulation. OTFS maps symbols to the DD domain, where the time-varying channel appears approximately quasi-static and can be estimated and equalized effectively. Fractional Doppler shifts, however, reduce DD sparsity and introduce modeling error when the BEM basis does not represent them adequately. The estimator of Section~\ref{sec:estimation-problem} operates in any linear transformation domain satisfying $\mathbf{D}^H \mathbf{D} = \mathbf{I}$ because the ZPDI condition of Section~\ref{sec:zero-pilot-interference} and the ZPPL condition of Section~\ref{subsec:complexity} do not depend on the observation domain. Setting $\mathbf{D}=\mathbf{I}$ places the estimator directly in the time domain, without loss of generality, and gives
\begin{equation}
\boxed{\hat{\boldsymbol{\gamma}}=\mathbf{R}_{pp}^{-1}\mathbf{t}_p=\left(\boldsymbol{\Psi}_{p}^{H}\boldsymbol{\Psi}_{p}\right)^{-1}\boldsymbol{\Psi}_{p}^{H}\mathbf{r}}.
\label{eq:optimal_solution_time_domain}
\end{equation}
This time-domain form is the one adopted throughout the OTFS specialization below.

\subsection{OTFS System Model}

Consider an OTFS delay--Doppler grid represented by a $K \times M$ matrix $\mathbf{S}$, where $K$ denotes the number of delay bins and $M$ denotes the number of Doppler bins. The total OTFS frame spans $N = KM$ time samples, where $N$ matches the observation interval length in the general formulation of Section~\ref{sec:channel-model}. In the reduced-CP variant of OTFS, one CP is prepended to the entire frame before transmission. After CP removal, the received signal in the DD domain can be expressed as
\begin{equation}
\mathbf{r} = \left(\mathbf{F}_{M} \otimes \mathbf{I}_{K}\right) \mathbf{H} \left(\mathbf{F}_{M}^{H} \otimes \mathbf{I}_{K}\right) \mathbf{s} + \mathbf{w},
\label{eq:otfs-system-model}
\end{equation}
where $\mathbf{F}_{M}$ is the $M$-point discrete Fourier transform (DFT) matrix, $\mathbf{H}$ is the time-domain doubly dispersive channel matrix, $\mathbf{s}$ is the $KM$-dimensional column vector obtained by column vectorization of the DD grid $\mathbf{S}$, and $\mathbf{w}$ is the additive noise vector. The system model \eqref{eq:otfs-system-model} has the same structure as the general model \eqref{eq:transform-domain-system-model} with the identification $\mathbf{D}=\left(\mathbf{F}_{M} \otimes \mathbf{I}_{K}\right)$ and $\mathbf{C}=  \left(\mathbf{F}_{M}^{H} \otimes \mathbf{I}_{K}\right)$.

\subsection{Single Pilot with Full Guard}
We consider an embedded single-pilot pattern in the DD grid, with guard symbols spanning all Doppler bins and the complete delay extent required by the channel support. This is the standard full-guard pilot configuration used in earlier OTFS work, including~\cite{raviteja2019}. We analyze whether this pattern guarantees pilot-only channel estimation under the ZPDI condition.

The pilot pattern uses a single pilot symbol $s_{p}\in\C$ placed at coordinates $(k_{p},m_{p})$ within the DD grid, where indices $k_p \in \{0, \ldots, K-1\}$ and $m_p \in \{0, \ldots, M-1\}$ specify the delay and Doppler bin positions, respectively. The pilot is surrounded by a rectangular guard region consisting of null symbols,

\begin{equation}
\mathbf{S}[k,m] = \begin{cases}
s_{p} & \text{if } k = k_{p}, \; m = m_{p} \\[0.5ex]
0 & \text{if } \begin{aligned}[t]
        &k_{p} - L + 1 \leq k \leq k_{p} + L - 1, \\
        &0 \leq m \leq M - 1, \\
        &(k,m) \neq (k_{p}, m_{p})
    \end{aligned} \\[0.5ex]
s_{u}[k,m] & \text{otherwise} \text{ (payload symbols)}
\end{cases}
\label{eq:otfs-pilot-pattern-expanded}
\end{equation}
where $L$ denotes the maximum channel delay spread expressed in delay bins, and $s_{u}[k,m]$ represents the payload symbols occupying positions outside the guard region. The guard region contains $(2L-1)M-1$ null symbols.

The pilot is typically positioned at $(k_{p},m_{p})=(\lfloor K/2\rfloor-1,\lfloor M/2\rfloor-1)$, near the geometric center of the delay--Doppler grid. The spectral-efficiency overhead of this pattern is $(2L-1)M/(KM) = (2L-1)/K$, namely, the fraction of cells reserved for channel estimation rather than data transmission. Because this overhead grows linearly with the channel-length parameter $L$, it remains moderate for channels with moderate delay spreads but becomes more costly as the delay spread increases. The applicability of the pattern therefore depends on the required balance between estimation accuracy and spectral efficiency.

\subsection{Analysis of Pilot--Data Orthogonality}

The single-pilot full-guard pattern of \eqref{eq:otfs-pilot-pattern-expanded} satisfies the zero pilot--data interference condition, so that it admits the exact, interference-free ZPDI estimator.

We evaluate the orthogonality condition in the equivalent time-domain representation obtained by setting $\mathbf{D}=\mathbf{I}$, as established above.

\begin{proposition}[Zero Pilot--Data Interference]
The single-pilot OTFS pattern defined in \eqref{eq:otfs-pilot-pattern-expanded} satisfies the zero pilot--data interference condition
\begin{equation}
\mathbf{R}_{pu}(\mathbf{s}_u)
=
\boldsymbol{\Psi}_{p}^{H}\boldsymbol{\Psi}_{u}(\mathbf{s}_u)
=\mathbf{0},\quad\forall\,\mathbf{s}_u,
\label{eq:otfs-orthogonality_requirement}
\end{equation}
where $\boldsymbol{\Psi}_{p} = \mathbf{Z}_{p}\mathbf{B}$ and $\boldsymbol{\Psi}_{u}(\mathbf{s}_u) = \mathbf{Z}_{u}\mathbf{B}$ are the pilot and payload observation matrices in the time-domain specialization $\mathbf{D}=\mathbf{I}$ introduced above.
\end{proposition}

\begin{IEEEproof}[Proof]
We verify the disjoint-support condition of Proposition~\ref{prop:disjoint-support}. A sufficient condition for \eqref{eq:otfs-orthogonality_requirement} is
\begin{equation}
\mathbf{z}_{p,\ell}^{*} \odot \mathbf{z}_{u,\nu} = \mathbf{0}, \quad \forall\,\ell,\nu\in\{0,1,\ldots,L-1\},
\label{eq:suff-cond-pilot-payload-orthogonality}
\end{equation}
where $\mathbf{z}_{p,\ell}$ and $\mathbf{z}_{u,\nu}$ are the $\ell$-th and $\nu$-th columns of $\mathbf{Z}_{p}$ and $\mathbf{Z}_{u}$. Vectorizing the DD grid \eqref{eq:otfs-pilot-pattern-expanded} places the unshifted pilot on $\mathcal{P}_{0}=\{k_{p}+qK\mid q=0,\ldots,M-1\}$, and the cyclic shift $\mathbf{z}_{p,\ell}=\mathbf{\Pi}^{\ell}\mathbf{z}_{p,0}$ moves it to $\mathcal{P}_{\ell}=\{(k_{p}+\ell+qK)\bmod N\mid q=0,\ldots,M-1\}$. The guard region of \eqref{eq:otfs-pilot-pattern-expanded} annihilates the unshifted payload on $\mathcal{G}=\bigcup_{q=0}^{M-1}\{k_{p}-L+1+qK,\ldots,k_{p}+L-1+qK\}$, and the shifted payload vector satisfies $\mathbf{z}_{u,\nu}[n]=\mathbf{x}_{u}[(n-\nu)\bmod N]$, so it vanishes wherever $(n-\nu)\bmod N\in\mathcal{G}$, that is, on a copy of $\mathcal{G}$ shifted by $\nu$, not on $\mathcal{G}$ itself.

Placing $k_{p}$ so that $[k_{p}-L+1,k_{p}+L-1]\subset\{0,\ldots,K-1\}$ excludes wrap-around for every admissible shift. For $n\in\mathcal{P}_{\ell}$, $n=k_{p}+\ell+qK$, so $n-\nu=k_{p}+(\ell-\nu)+qK$ with $\ell-\nu\in\{-(L-1),\ldots,L-1\}$, which always lies in $\{k_{p}-L+1+qK,\ldots,k_{p}+L-1+qK\}\subseteq\mathcal{G}$; hence $\mathbf{z}_{u,\nu}[n]=\mathbf{x}_{u}[n-\nu]=0$. Every nonzero entry of $\mathbf{z}_{p,\ell}$ therefore falls on a position where $\mathbf{z}_{u,\nu}$ vanishes, for every $\ell,\nu\in\{0,\ldots,L-1\}$, which establishes \eqref{eq:suff-cond-pilot-payload-orthogonality} and, by Proposition~\ref{prop:disjoint-support}, \eqref{eq:otfs-orthogonality_requirement}.
\end{IEEEproof}

The orthogonality condition \eqref{eq:suff-cond-pilot-payload-orthogonality} imposes the constraint $K>2L-2$. This requirement is necessary and sufficient to prevent wrap-around interference in the cyclic permutations induced by the cyclic prefix under the adopted guard pattern. Indeed, the upper guard boundary $(k_{p}+L-1)$ must not interfere with the lower boundary $(k_{p}-L+1)$ after modulo-$K$ wrapping, which requires $(k_{p}+L-1)-(k_{p}-L+1)<K$, or equivalently $2L-2<K$. Additionally, proper placement of the pilot with guards on both sides requires $L-1\leq k_{p}\leq K-L$, which is possible only when $K>2L-2$.

\subsection{Pilot--Pilot Leakage}

The adopted full-guard pattern satisfies ZPDI. We next determine whether it also satisfies the ZPPL condition in \eqref{eq:zero_pilot_leakage_condition}.
\begin{proposition}[ZPPL for the OTFS Full-Guard Pattern]
For the single-pilot full-guard pattern, the shifted pilot vectors have
disjoint supports across different delay taps, i.e.,
\begin{equation}
\mathbf{z}_{p,i}^{*}\odot\mathbf{z}_{p,j}=\mathbf{0},
\quad i\neq j,
\label{eq:otfs-pilot-disjoint-support}
\end{equation}
and each shifted pilot vector satisfies
\begin{equation}
\|\mathbf{z}_{p,i}\|^{2}=\mathcal{E}_{p},
\quad i=0,\ldots,L-1,
\label{eq:otfs-pilot-energy}
\end{equation}
where $\mathcal{E}_{p} = |s_{p}|^{2}$ denotes the pilot energy.
Consequently, the pattern satisfies ZPPL for any BEM basis
$\boldsymbol{\Phi}$.
\end{proposition}

\begin{IEEEproof}
For $i \neq j$, the supports of the shifted pilot vectors are
\[
\begin{aligned}
\mathcal{P}_i &= \{k_p+i+qK \mid q=0,\ldots,M-1\}, \\
\mathcal{P}_j &= \{k_p+j+qK \mid q=0,\ldots,M-1\}.
\end{aligned}
\]
They do not overlap because $i\neq j$ and, under the condition $K>2L-2$,
no modulo-$N$ wrap-around occurs for $i,j\in\{0,\ldots,L-1\}$. Hence no
sample can be nonzero in both $\mathbf{z}_{p,i}$ and $\mathbf{z}_{p,j}$,
which proves \eqref{eq:otfs-pilot-disjoint-support}. For $i = j$, each
pilot vector $\mathbf{z}_{p,i}$ contains $M$ nonzero elements at positions
$k_p + i + qK$ for $q = 0, \ldots, M-1$, each with magnitude
$|s_p|/\sqrt{M}$ due to the inverse symplectic finite Fourier transform (ISFFT) normalization. Thus
\begin{equation}
\begin{split}
\|\mathbf{z}_{p,i}\|^{2} &= \sum_{q=0}^{M-1}\left|\frac{s_{p}}{\sqrt{M}}\right|^{2} = \sum_{q=0}^{M-1}\frac{|s_{p}|^{2}}{M} \\
&= M \cdot \frac{|s_{p}|^{2}}{M} = |s_{p}|^{2} = \mathcal{E}_{p}
\end{split}
\end{equation}
which proves \eqref{eq:otfs-pilot-energy}. Finally, the cross-Gram identity
\eqref{eq:hadamard_product_form} in the proof of
Proposition~\ref{prop:disjoint-support}, applied with $\mathbf{z}_{u,m}$
replaced by $\mathbf{z}_{p,j}$, maps the disjoint supports
\eqref{eq:otfs-pilot-disjoint-support} directly to
$\boldsymbol{\Psi}_{p,i}^{H}\boldsymbol{\Psi}_{p,j}=\mathbf{0}$ for $i\neq j$,
independently of the BEM basis.
\end{IEEEproof}

The proposition establishes elementwise disjoint support, which is the condition
required for ZPPL and is stronger than ordinary vector orthogonality. It therefore verifies
\eqref{eq:zero_pilot_leakage_condition} for the OTFS full-guard pattern. By the general block-diagonal result
\eqref{eq:block_diagonal_normal_equations}, the pilot Gram matrix
$\mathbf{R}_{pp}=\boldsymbol{\Psi}_{p}^{H}\boldsymbol{\Psi}_{p}$, with
$\boldsymbol{\Psi}_{p}=\mathbf{Z}_{p}\mathbf{B}$, is therefore block diagonal.
Since each delay tap excites $M$ time-domain samples of equal energy
$\mathcal{E}_{p}/M$, it takes the explicit form
\begin{equation}
\mathbf{R}_{pp}
= \frac{\mathcal{E}_{p}}{M}\,\blkdiag\!\left(\mathbf{T}_0,\mathbf{T}_1,\ldots,\mathbf{T}_{L-1}\right),
\label{eq:otfs-rpp-blockdiag}
\end{equation}
where $\mathbf{T}_{\ell}\in\C^{Q\times Q}$ is the Gram matrix of the $Q$ BEM basis functions sampled at the $M$ pilot positions of delay tap $\ell$. Thus $\mathbf{T}_{\ell}$ is the OTFS specialization of the general diagonal block $\mathbf{R}_{\ell\ell}$ of Section~\ref{subsec:complexity}, with $\mathbf{R}_{\ell\ell}=(\mathcal{E}_{p}/M)\,\mathbf{T}_{\ell}$. The estimator then separates into the $L$ independent per-tap solves that specialize \eqref{eq:block_diagonal_estimation},
\begin{equation} 
    \boxed{
    \hat{\boldsymbol{\gamma}}_{\ell} = \frac{M}{\mathcal{E}_{p}}\,\mathbf{T}_{\ell}^{-1}\boldsymbol{\Psi}_{p,\ell}^{H}\mathbf{r}, \quad \ell = 0, 1, \ldots, L-1,
    }
    \label{eq:otfs-per-tap-estimator}
\end{equation}
where $\boldsymbol{\Psi}_{p,\ell}$ is the $\ell$-th block column of $\boldsymbol{\Psi}_{p}$.

The per-tap form makes the identifiability condition of Section~\ref{sec:identifiability} concrete for OTFS. Each $\mathbf{T}_{\ell}$ is the Gram matrix of $Q$ basis functions observed at only $M$ pilot positions, so its rank cannot exceed $M$, and the inverse in \eqref{eq:otfs-per-tap-estimator} exists only if
\begin{equation}
Q \le M .
\label{eq:otfs-identifiability}
\end{equation}
This is the per-tap instance of the general counting condition \eqref{eq:counting-condition}: the single-pilot full-guard pattern supplies $N_{\mathrm{obs}}=M$ independent observations per delay tap, so the admissible BEM order is bounded by the number of Doppler bins. 

\subsection{GCE-BEM Analysis}
\label{sec:otfs-gce-bem-analysis}

The full-guard pattern guarantees pilot--data orthogonality for any BEM choice, but the basis functions still determine the pilot Gram matrix $\mathbf{R}_{pp}=\boldsymbol{\Psi}_{p}^{H}\boldsymbol{\Psi}_{p}$, which governs numerical conditioning and inversion complexity. Since $\mathbf{R}_{pp}$ is block diagonal, the analysis reduces to the per-tap blocks $\mathbf{T}_{\ell}$, the Gram matrix of the basis functions restricted to the pilot support $\mathcal{P}_{\ell}$.

We adopt the GCE-BEM basis $\phi_{q}[n]=\tfrac{1}{\sqrt{N}}\exp(j\omega_{q}n)$, for $n=0,\ldots,N-1$ and $q=0,\ldots,Q-1$, with frequencies $\omega_{q}=\tfrac{2\pi}{NR}\big(q-\lceil (Q-1)/2\rceil\big)$. The oversampling factor $R\ge1$ sets the Doppler-grid spacing to $2\pi/(NR)$; $R=1$ recovers the conventional CE-BEM on the DFT grid $2\pi/N$, while larger $R$ refines the Doppler grid and better resolves fractional Doppler components.

For the single-pilot pattern, the pilot positions follow $n=k_{p}+\ell+qK$, $q=0,\ldots,M-1$, and direct summation of $\mathbf{T}_{\ell}[a,b]=\sum_{n\in\mathcal{P}_{\ell}}\phi_{a}^{*}[n]\phi_{b}[n]$ over this geometric sequence gives the closed-form entries
\begin{equation}
\begin{aligned}
\mathbf{T}_\ell[a,b] &= \frac{1}{K} \exp\left(j\frac{2\pi(b-a)(k_{p}+\ell)}{NR}\right) \\
&\quad \times \exp\left(j\frac{\pi(M-1)(b-a)}{MR}\right)
\mathcal{D}_{M}\left(\frac{\pi(b-a)}{MR}\right),
\end{aligned}
\label{eq:gce-bem-t-elements}
\end{equation}
where $a,b\in\{0,\ldots,Q-1\}$ and $\mathcal{D}_{M}(x)=\sin(Mx)/(M\sin x)$ is the Dirichlet kernel. Equation~\eqref{eq:gce-bem-t-elements} shows that $\mathbf{T}_{\ell}$ is Hermitian Toeplitz, with diagonal $\mathbf{T}_{\ell}[a,a]=1/K$ since $\mathcal{D}_{M}(0)=1$; for $R>1$ the Dirichlet kernel populates the off-diagonal entries, which vanish only when $b-a$ is a nonzero multiple of $R$ with $(b-a)\bmod MR\neq0$. Since $\mathbf{T}_{\ell}$ is generated by the $M$ pilot samples of $\mathcal{P}_{\ell}$, $\rank(\mathbf{T}_{\ell})\le M$, confirming for the GCE-BEM the per-tap identifiability ceiling $Q\le M$ of \eqref{eq:otfs-identifiability}: the per-tap problem is exactly determined at $Q=M$, overdetermined and better conditioned at $Q<M$, and rank deficient at $Q>M$. The oversampling factor therefore trades Doppler resolution against conditioning: a larger $R$ reduces the fractional-Doppler modeling error but, through the same loss of basis orthogonality, raises $\kappa(\mathbf{R}_{pp})$ and amplifies the noise at high SNR. 

When $R=1$, the Dirichlet kernel is orthogonal for $Q\le M$, so every per-tap block collapses to $\mathbf{T}_{\ell}=\tfrac{1}{K}\mathbf{I}_{Q}$ and the pilot Gram matrix reduces to the scaled identity $\mathbf{R}_{pp}=(\mathcal{E}_{p}/N)\mathbf{I}_{QL}$, the diagonal case of Section~\ref{subsec:complexity} with ideal conditioning $\kappa(\mathbf{R}_{pp})=1$ and no matrix inversion. The per-tap estimator \eqref{eq:otfs-per-tap-estimator} then reduces to the scaled matched filter
\begin{equation}
    \boxed{
    \hat{\boldsymbol{\gamma}}_{\ell} = \frac{N}{\mathcal{E}_{p}} \boldsymbol{\Psi}_{p,\ell}^{H}\mathbf{r}, \quad \ell = 0, 1, \ldots, L-1.
    }
    \label{eq:ce-bem-estimator}
\end{equation}

\section{Numerical Illustrations and End-to-End Results}
\label{sec:numerical_examples}

Propositions~\ref{prop:optimal-estimator} and~\ref{prop:disjoint-support} are analytical statements that hold for every admissible payload. The numerical examples instead illustrate their consequences for OTFS and assess the selected GCE-BEM in the presence of model mismatch.

The simulations use an OTFS grid with $K=128$ delay bins and $M=16$ Doppler bins, for a total of $KM=2048$ cells. The channel follows the Third Generation Partnership Project (3GPP) tapped-delay-line B (TDL-B) model~\cite{tr38901}, with a root-mean-square (RMS) delay spread of $300$~ns and Jakes' Doppler spectrum. We set the carrier frequency to $4$~GHz and the subcarrier spacing to $\Delta f=15$~kHz, consistent with fifth-generation New Radio numerology. For this setting, the discrete-time channel has length $L=4$.

A single pilot is placed at
$(k_p,m_p)=(\lfloor K/2\rfloor-1,\lfloor M/2\rfloor-1)$, using the full-guard
pattern of Section~\ref{sec:otfs-examples}. The frame then contains
$(2L-1)M=112$ pilot-plus-guard cells and $1936$ data cells. The corresponding
pilot overhead is
\[
\frac{(2L-1)M}{KM}=\frac{7}{128}\approx 5.5\%.
\]

The single nonzero pilot carries the average energy that quadrature phase-shift keying (QPSK) data symbols
would transmit over the same $(2L-1)M$ cells. All estimators use the same
transmitted frame, pilot energy, and noise variance. Table~\ref{tab:system-parameters}
collects the remaining simulation parameters.

\begin{table}[t]
    \centering
    \caption{Simulation parameters used in the OTFS numerical examples.}
    \label{tab:system-parameters}
    \begin{tabular}{|l|c|}
        \hline
        \textbf{Parameter} & \textbf{Value} \\
        \hline
        Carrier frequency & $4$ GHz \\
        Subcarrier spacing & $15$ kHz \\
        Number of delay bins ($K$) & $128$ \\
        Number of Doppler bins ($M$) & $16$ \\
        Pilot-plus-guard cells & $112$ \\
        Modulation  & QPSK \\
        Channel model & TDL-B \\
        RMS delay spread & $300$ ns \\
        Channel length ($L$) & $4$ \\
        Oversampling factor ($R$) & $2$ \\
        Mobility values & $125$, $500$ km/h \\
        GCE-BEM order ($Q$) & $3$ at $125$; $9$ at $500$ km/h \\
        CE-BEM stress order ($Q$) & $15$ at $125$ km/h \\
        \hline
    \end{tabular}
\end{table}

We evaluate the bit error rate (BER) and the channel-estimation MSE, defined as
\[
\text{MSE}=\frac{1}{NL}\,\E\!\left[\|\mathbf{G}-\hat{\mathbf{G}}\|_F^2\right],
\]
where $\mathbf{G}$ is the true sample-rate channel coefficient matrix and
$\hat{\mathbf{G}}$ is its estimate. The expectation is taken over the channel,
payload, and noise realizations.

\subsection{Illustration of the ZPDI Consequences}
\label{sec:zpdi-mechanism-validation}

We first examine three consequences of the analysis before presenting the end-to-end BER and MSE comparisons: the ZPDI decoupling of Proposition~\ref{prop:optimal-estimator},
the error decomposition in \eqref{eq:pdi-contaminated-pilot-estimate}, and the
Gram-matrix properties developed in Section~\ref{sec:gram-properties}. Unless stated otherwise, the setup
is the $500$~km/h operating point of Table~\ref{tab:system-parameters}
($K=128$, $M=16$, $L=4$, GCE-BEM with $R=2$ and $Q=9$); only the channel model
is replaced, as described next. To remove the confounding effect of BEM
modeling error, the channel is synthesized \emph{exactly} in the BEM subspace,
$\mathbf{G}=\boldsymbol{\Phi}\boldsymbol{\Gamma}$ with
$\operatorname{vec}(\boldsymbol{\Gamma})\sim\mathcal{CN}(\mathbf{0},\mathbf{I}_{QL})$,
so that $\mathbf{E}=\mathbf{0}$ and the only remaining error sources are the two
terms of \eqref{eq:pdi-contaminated-pilot-estimate}: the payload-dependent bias
and the thermal noise shaped by $\mathbf{R}_{pp}^{-1}$. Expectations are
estimated over $400$ independent channel, payload, and noise realizations
($200$ for the pilot-power sweep). The non-ZPDI configuration is the same
single-pilot pattern with the delay guard reduced to half-width $d=1<L-1$ while
the Doppler guard remains full, so that a channel shift maps the pilot onto data
cells and reintroduces $\mathbf{R}_{pu}(\mathbf{s}_u)$.

\emph{Payload decoupling.}
Figure~\ref{fig:zpdi-mechanism}(a) reports the worst-case normalized
pilot--data leakage
\begin{equation}
\epsilon_{\mathrm{PDI}}^{\max}
=\max_{j}\;\frac{\big\|\boldsymbol{\Psi}_p^H\boldsymbol{\Psi}_u(\mathbf{e}_j)\big\|_F}
{\|\mathbf{R}_{pp}\|_F},
\label{eq:eps-pdi-metric}
\end{equation}
where $\mathbf{e}_j$ activates a single payload cell. By linearity of
$\boldsymbol{\Psi}_u(\cdot)$, vanishing leakage for every elementary
$\mathbf{e}_j$ certifies it for \emph{every} admissible payload, making
\eqref{eq:eps-pdi-metric} a deterministic test of
\eqref{eq:orthogonality-condition} rather than a sampled average. Within the
tested family of symmetric delay guards
$d\in\{0,\ldots,L-1\}$ with full Doppler guard, only the full guard $d=L-1$
drives $\epsilon_{\mathrm{PDI}}^{\max}$ to machine precision ($<10^{-15}$);
deficient guards $d\in\{0,1,2\}$ leave
$\epsilon_{\mathrm{PDI}}^{\max}\in\{8.2,6.7,4.7\}\times10^{-2}$. This
illustrates Proposition~\ref{prop:optimal-estimator}: the full
guard is the member of this family satisfying the disjoint-support rule of
Proposition~\ref{prop:disjoint-support}, hence making
$\mathbf{R}_{pu}(\mathbf{s}_u)=\mathbf{0}$ for every admissible payload. The
operator-level condition \eqref{eq:orthogonality-condition} does not, in
general, single out a specific guard width, and can also be met by algebraic
cancellation without disjoint supports.

\emph{Error decomposition and the interference floor.}
Figure~\ref{fig:zpdi-mechanism}(b) plots the measured channel-estimation normalized MSE (NMSE) of
the matched-pilot estimator \eqref{eq:optimal_solution_zero_interference} versus
SNR (markers, Monte Carlo) against the closed-form prediction (solid lines)
obtained from \eqref{eq:pdi-contaminated-pilot-estimate},
\begin{equation}
\mathrm{NMSE}(\mathrm{SNR})
=\frac{\E\big\|\mathbf{R}_{pp}^{-1}\mathbf{R}_{pu}(\mathbf{s}_u)\boldsymbol{\gamma}\big\|_{\mathbf{W}}^2
+\sigma_w^2\,\tr\!\big(\mathbf{W}\mathbf{R}_{pp}^{-1}\big)}
{\E\|\boldsymbol{\gamma}\|_{\mathbf{W}}^2},
\label{eq:nmse-decomposition}
\end{equation}
with the channel-domain weighting
$\mathbf{W}=\mathbf{I}_L\otimes\boldsymbol{\Phi}^H\boldsymbol{\Phi}$ and
$\|\mathbf{x}\|_{\mathbf{W}}^2=\mathbf{x}^H\mathbf{W}\mathbf{x}$, so that
$\|\hat{\mathbf{G}}-\mathbf{G}\|_F^2=\|\hat{\boldsymbol{\gamma}}-\boldsymbol{\gamma}\|_{\mathbf{W}}^2$
under the exact-BEM channel. The Monte Carlo NMSE follows the prediction to
within the sampling scatter across the entire range. For the ZPDI full-guard
pattern the bias term vanishes and the NMSE coincides with the noise-only curve
$\sigma_w^2\tr(\mathbf{W}\mathbf{R}_{pp}^{-1})$, decreasing without a floor. For
the non-ZPDI pattern the bias term is constant in SNR and produces the high-SNR
floor predicted by \eqref{eq:pdi-contaminated-pilot-estimate}. Here the pilot
carries the QPSK-matched energy of Table~\ref{tab:system-parameters}, i.e., unit
pilot-to-data power ratio. Figure~\ref{fig:zpdi-floor-id}(a) then fixes the data
power and sweeps the pilot energy: the floor, evaluated as the SNR-independent
bias term of \eqref{eq:nmse-decomposition}, attenuates as $1/\mathcal{E}_p$ but
does not vanish at any finite pilot-to-data ratio, so the interference cannot be
removed by pilot power alone.

\emph{Identifiability, conditioning, and structure.}
The counting ceiling $Q\le M$ of \eqref{eq:otfs-identifiability} is
basis-independent, so to exhibit it without the $R=2$ conditioning confound
Figure~\ref{fig:zpdi-floor-id}(b) sweeps the BEM order for the full-guard
pattern with the well-conditioned CE-BEM ($R=1$), for which
$\mathbf{R}_{pp}=(\mathcal{E}_p/N)\mathbf{I}_{QL}$ while $Q\le M$. The smallest
eigenvalue $\lambda_{\min}(\mathbf{R}_{pp})$ stays at $\mathcal{E}_p/N$ with zero
rank deficit for all $Q\le M$: at $Q=M$ the per-tap problem is exactly
determined but still full rank. For $Q>M$ the rank deficit opens as $L(Q-M)$ and
$\lambda_{\min}$ drops to zero, so structural singularity begins at $Q=M+1$,
illustrating the ceiling \eqref{eq:otfs-identifiability}; the numerical rank uses
the tolerance $QL\,\varepsilon\,\|\mathbf{R}_{pp}\|$.
Table~\ref{tab:gram-diagnostics} collects the remaining deterministic checks at
the $Q=9$, $R=2$ operating point: the off-tap Gram energy is zero to machine
precision, verifying the ZPPL block-diagonalization of
\eqref{eq:otfs-rpp-blockdiag}, and the per-tap blocks match the closed-form
Toeplitz expression \eqref{eq:gce-bem-t-elements} to $3\times10^{-16}$. The
condition number $\kappa(\mathbf{R}_{pp})\approx7.0\times10^{5}$ enters the
predicted noise term $\sigma_w^2\tr(\mathbf{W}\mathbf{R}_{pp}^{-1})$ of
\eqref{eq:nmse-decomposition} and quantifies the noise enhancement of this
$R=2$ oversampled basis discussed in Section~\ref{sec:conditioning}; the CE-BEM
limit $R=1$ gives $\kappa(\mathbf{R}_{pp})=1$.

\begin{figure}[t]
\centering
\includegraphics[width=\linewidth]{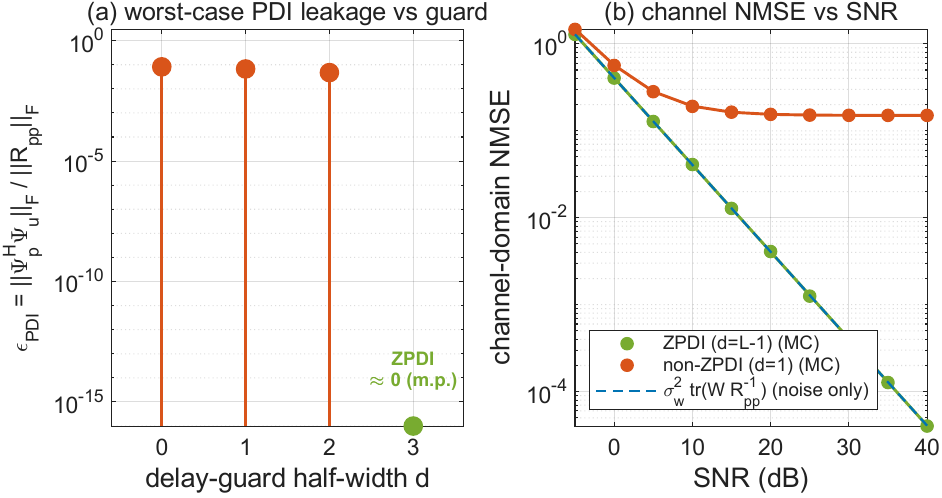}
\caption{Illustration of the ZPDI consequences under an exact-BEM channel
($\mathbf{E}=\mathbf{0}$), at the $500$~km/h operating point of
Table~\ref{tab:system-parameters} ($R=2$, $Q=9$). (a) Worst-case normalized
pilot--data leakage $\epsilon_{\mathrm{PDI}}^{\max}$ of
\eqref{eq:eps-pdi-metric} over all elementary payloads versus delay-guard
half-width $d$: it collapses to machine precision only at the full guard
$d=L-1$, as predicted by the disjoint-support rule of
Proposition~\ref{prop:disjoint-support}. (b) Channel-estimation NMSE versus SNR
for the ZPDI ($d=L-1$) and non-ZPDI ($d=1$) patterns: markers are the measured
estimator NMSE and solid lines the closed-form prediction
\eqref{eq:nmse-decomposition}; the dashed line is the noise-only term. The
non-ZPDI floor is the payload-dependent bias of
\eqref{eq:pdi-contaminated-pilot-estimate}.}
\label{fig:zpdi-mechanism}
\end{figure}

\begin{figure}[t]
\centering
\includegraphics[width=\linewidth]{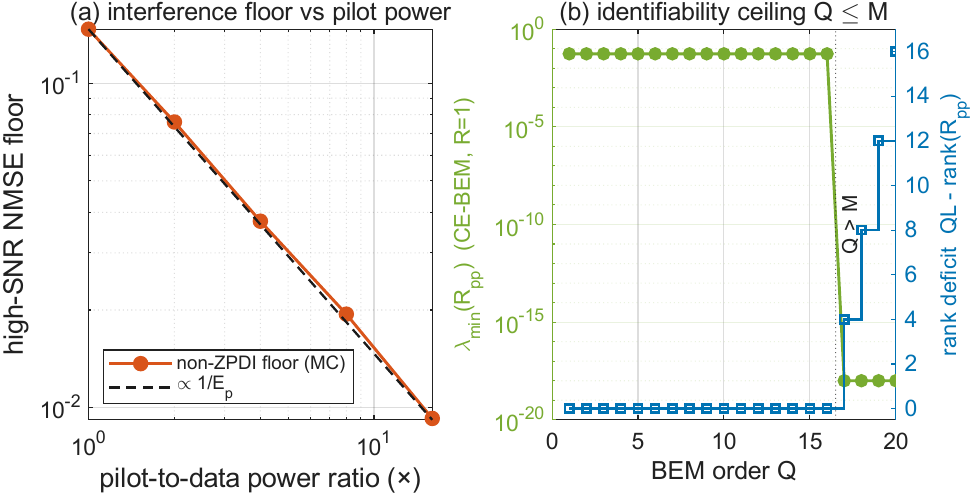}
\caption{(a) High-SNR NMSE floor of the non-ZPDI pattern versus pilot-to-data
power ratio, at fixed data power: the floor attenuates as $1/\mathcal{E}_p$
(dashed) but stays positive at every finite ratio, as anticipated after
\eqref{eq:pdi-contaminated-pilot-estimate}. (b) Smallest eigenvalue (left,
CE-BEM $R=1$) and rank deficit $QL-\rank(\mathbf{R}_{pp})$ (right) of the
full-guard Gram matrix versus BEM order $Q$: $\mathbf{R}_{pp}$ is full rank for
$Q\le M$ and rank-deficient by $L(Q-M)$ for $Q>M$, illustrating the
identifiability ceiling \eqref{eq:otfs-identifiability}.}
\label{fig:zpdi-floor-id}
\end{figure}

\begin{table}[t]
\centering
\caption{Gram-matrix diagnostics for the OTFS full-guard pattern, illustrating the
analysis of Section~\ref{sec:gram-properties}. Rows are evaluated at the $Q=9$,
$R=2$ operating point unless stated; the identifiability row uses the CE-BEM
($R=1$).}
\label{tab:gram-diagnostics}
\begin{tabular}{|l|c|c|}
\hline
\textbf{Quantity} & \textbf{Measured} & \textbf{Predicted} \\
\hline
$\epsilon_{\mathrm{PDI}}^{\max}$ (full guard) & $<10^{-15}$ & $0$ \eqref{eq:orthogonality-condition} \\
Off-tap Gram energy (rel.) & $<10^{-15}$ & $0$ \eqref{eq:zero_pilot_leakage_condition} \\
Toeplitz closed-form error & $2.8\times10^{-16}$ & $0$ \eqref{eq:gce-bem-t-elements} \\
$\rank(\mathbf{R}_{pp})$, $Q=9$ & $36$ & $QL=36$ \\
Rank deficit at $Q{=}M{+}1$ ($R{=}1$) & $L=4$ & $L(Q{-}M)$ \eqref{eq:otfs-identifiability} \\
$\kappa(\mathbf{R}_{pp})$, $R{=}2$ / $R{=}1$ & $7.0\times10^{5}$ / $1$ & --- \\
\hline
\end{tabular}
\end{table}

\subsection{ZPDI Estimator}

The standalone ZPDI receiver uses a time-domain linear minimum mean-squared
error (TD-LMMSE) equalizer built from the estimated channel, compared against
perfect CSI, the estimator of Raviteja \textit{et al.}~\cite{raviteja2019}, and
the pilot-only initialization of the iterative BEM receiver of Liu
\textit{et al.}~\cite{liu2022}.

Since~\cite{raviteja2019} returns the DD-domain matrix
$\widehat{H}_{\mathrm{dd}} \in \C^{K\times M}$ rather than BEM coefficients, we
recover slot-rate coefficients per delay tap $\ell$ by an inverse DFT along
Doppler,
\begin{equation}
    g_\ell[q] = \sum_{m=0}^{M-1}\widehat{H}_{\mathrm{dd}}[k_p+\ell,m]\,
    e^{j2\pi qm/M}, \quad q = 0,\ldots,M-1,
\end{equation}
and linearly interpolate to the sample-rate vector $\mathbf{g}_\ell$ used by the
TD-LMMSE equalizer and in the MSE evaluation.

The simulations consider $v=125$ and $500$~km/h. The BEM order is selected within
\begin{equation}
Q_{\min} \;\triangleq\; 2\left\lceil R \nu_{\max} \frac{M}{\Delta f} \right\rceil + 1
\;\le\; Q \;\le\; M,
\label{eq:q-selection}
\end{equation}
where $T_{\mathrm{obs}}=M/\Delta f$ is the frame observation time (the $M$
Doppler symbols, each of duration $1/\Delta f$), so that $\lceil R\nu_{\max}
T_{\mathrm{obs}}\rceil$ counts the GCE-BEM Doppler-grid points needed on each side
of the zero-Doppler term; hence $Q_{\min}$ is the smallest odd GCE-BEM order
covering $[-\nu_{\max},\nu_{\max}]$ and the upper bound is the identifiability
limit in \eqref{eq:otfs-identifiability}. Increasing $Q$ reduces modeling error but
enlarges the coefficient vector and Gram matrix; the MSE decomposition
in~\cite{romano2026} confirms this trade-off (noise term grows linearly with
$Q$, mismatch term decreases). The Liu receiver~\cite{liu2022} uses the $R=1$
GCE-BEM (CE-BEM special case) with $Q_S=15$ for initialization and the $R=2$
GCE-BEM with $Q_L=9$ for data-aided refinement; since $M=16$, $Q_S=15$ is the
largest admissible odd order in \eqref{eq:q-selection}.

Figs.~\ref{fig:v125-ber} and~\ref{fig:v500-ber} report BER versus SNR. At
$v=125$~km/h, GCE-BEM ($Q=3$) stays close to perfect CSI, whereas CE-BEM floors
at moderate-to-high SNR even at its largest admissible order $Q=15$. At
$v=500$~km/h, the ZPDI estimator ($R=2$, $Q=9$) is within about $2$~dB of
perfect CSI at high SNR. ZPDI removes pilot--data interference but not basis
mismatch, so the BER still depends on the selected BEM. The Liu initialization
($R=1$, $Q_S=15$) floors at high SNR despite the largest admissible odd CE-BEM
order, consistent with residual basis mismatch. The estimator
of~\cite{raviteja2019} lies between the two and also floors at $v=500$~km/h,
where fractional Doppler distorts the reconstructed pilot response.

\begin{figure*}[t]
\centering
\subfloat[BER]{
\includegraphics[width=0.45\textwidth]{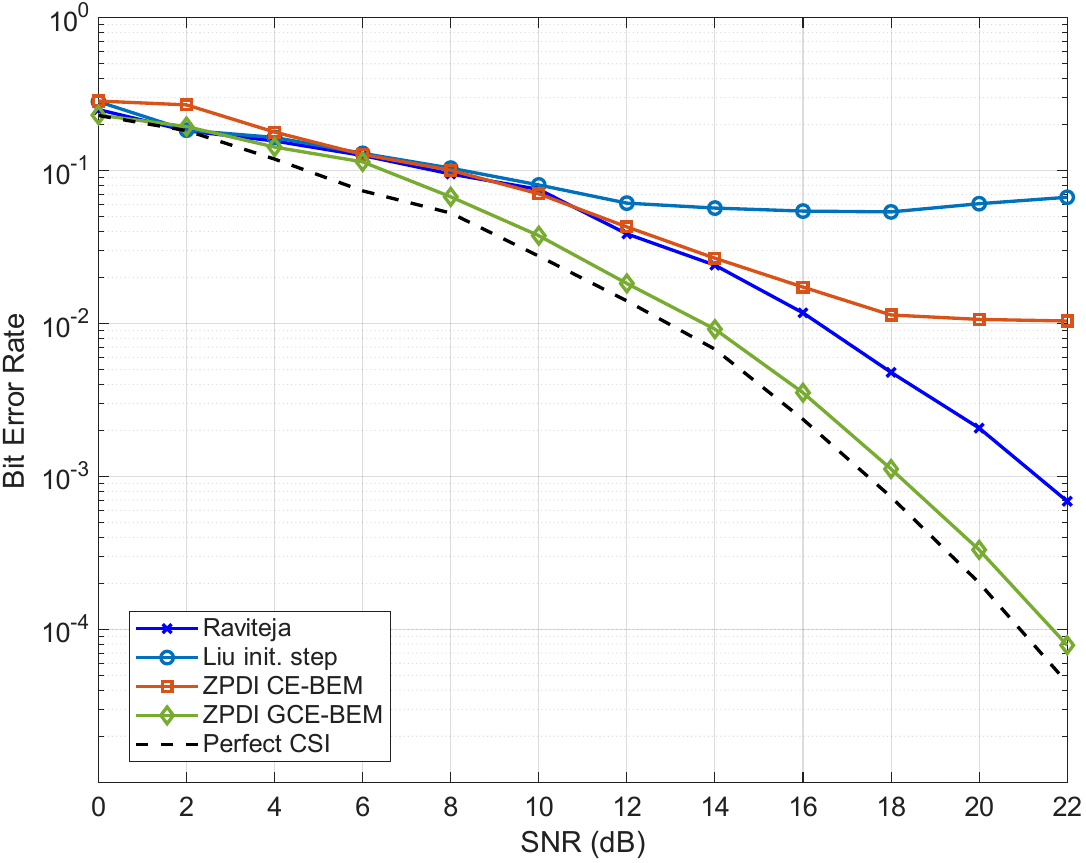}
\label{fig:v125-ber}
}
\hfill
\subfloat[MSE]{
\includegraphics[width=0.45\textwidth]{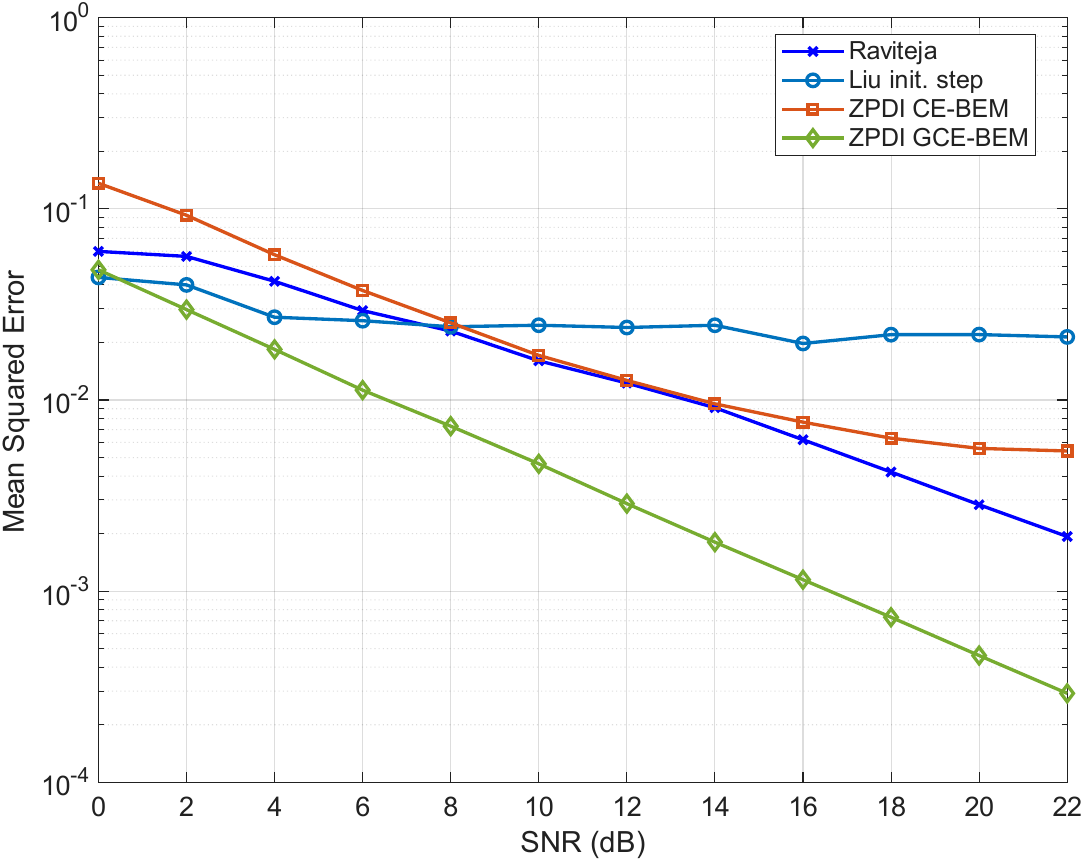}
\label{fig:v125-mse}
} 
\caption{(a) BER and (b) MSE versus SNR at $v=125$~km/h for the setup in Table~\ref{tab:system-parameters}. The curves show the ZPDI estimator with GCE-BEM, the CE-BEM stress case, the initial pilot-only step of the Liu receiver~\cite{liu2022} with $(R,Q_S)=(1,15)$, the estimator of~\cite{raviteja2019}, and perfect CSI. The CE-BEM stress case and the Liu initialization retain high-SNR floors at the largest admissible odd CE-BEM order, which is consistent with BEM mismatch. The ZPDI estimator contains no pilot--data interference by construction; see \eqref{eq:pdi-contaminated-pilot-estimate}.}
\label{fig:v125-results}
\end{figure*}

Figs.~\ref{fig:v125-mse} and~\ref{fig:v500-mse} show the corresponding MSE. The
CE-BEM and Liu-initialization floors account for their BER floors, and the
estimator of~\cite{raviteja2019} sits above the ZPDI GCE-BEM curve with a clear
floor at $v=500$~km/h. The ZPDI GCE-BEM MSE keeps decreasing with SNR in both
scenarios, with no visible interference-induced floor.

\begin{figure*}[t]
\centering
\subfloat[BER]{
\includegraphics[width=0.45\textwidth]{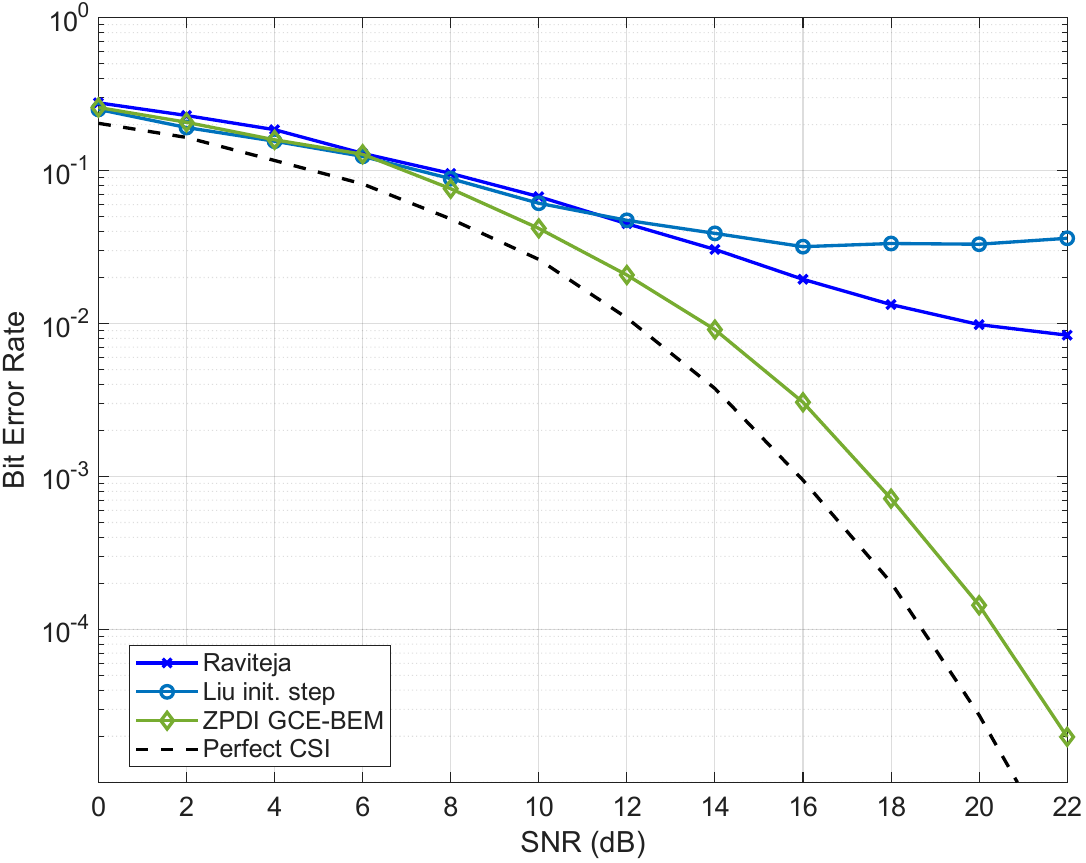}
\label{fig:v500-ber}
}
\hfill
\subfloat[MSE]{
\includegraphics[width=0.45\textwidth]{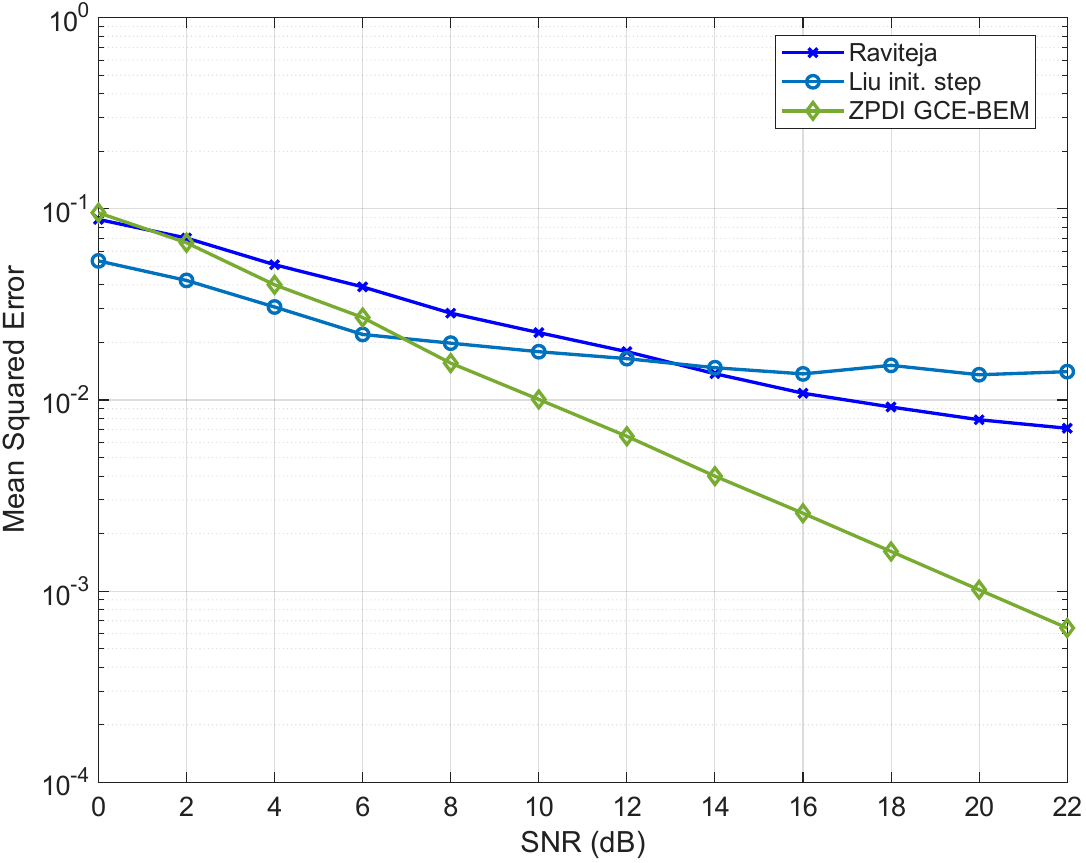}
\label{fig:v500-mse}
}
\caption{(a) BER and (b) MSE versus SNR at $v=500$~km/h. The estimators are the same as in Fig.~\ref{fig:v125-results}, except that the CE-BEM stress case is omitted; the Liu initialization uses $(R,Q_S)=(1,15)$. The reference floors are consistent with BEM or time-domain reconstruction mismatch. The ZPDI estimator has no visible floor over the simulated SNR range; see \eqref{eq:pdi-contaminated-pilot-estimate}.}
\label{fig:v500-results}
\end{figure*}

\subsection{ZPDI Initialization of a Decision-Directed Receiver}

The pilot-only ZPDI estimate can initialize an iterative receiver: tentative
payload decisions define a data-aided LS update that approximates the
genie-aided estimator in \eqref{eq:ls_solution_known_symbols}.

We evaluate this decision-directed receiver at $v=500$~km/h. Both branches share
the same OTFS frame, TDL-B channel, pilot energy, full-guard pattern, and
TD-LMMSE detector as before, so the comparison isolates the pilot-only
initialization. The ZPDI-initialized branch uses $(R,Q)=(2,9)$ in both stages;
the Liu receiver~\cite{liu2022} uses $(R,Q_S)=(1,15)$ for initialization and
$(R,Q_L)=(2,9)$ for refinement, varying the number of refinements.

Algorithm~\ref{alg:dd-receiver} gives the receiver steps, where
$\mathsf{Dec}(\mathbf{r};\hat{\mathbf{G}})$ is the TD-LMMSE detector followed by
hard QPSK slicing. At iteration $i$, the detected payload defines the
data-aided LS update
\begin{equation}
\hat{\boldsymbol{\gamma}}^{(i+1)}
=
\left(
\boldsymbol{\Psi}^{H}(\hat{\mathbf{s}}_u^{(i)})
\boldsymbol{\Psi}(\hat{\mathbf{s}}_u^{(i)})
\right)^{-1}
\boldsymbol{\Psi}^{H}(\hat{\mathbf{s}}_u^{(i)})\mathbf{r}.
\label{eq:decision-directed-ls-update}
\end{equation}
The iterations stop at $I_{\max}$ or when the update metric
$\epsilon_i^{\mathrm{upd}}\triangleq (NL)^{-1}\|\hat{\mathbf{G}}^{(i+1)}-\hat{\mathbf{G}}^{(i)}\|_F^2$
falls below $\epsilon_{\mathrm{stop}}$, for example, $10^{-2}$.

\begin{algorithm}[t]
\caption{Decision-directed receiver with ZPDI initialization}
\label{alg:dd-receiver}
\footnotesize
\begin{algorithmic}[1]
\REQUIRE $\mathbf{r}$, $\mathbf{s}_p$, $\mathbf{P}_p$, $\mathbf{P}_u$,
$\boldsymbol{\Phi}$, $\mathbf{W}_p$, $\mathsf{Dec}$, $I_{\max}$,
$\epsilon_{\mathrm{stop}}$.
\ENSURE $\hat{\mathbf{s}}_u$, $\hat{\mathbf{G}}$.
\STATE $\hat{\boldsymbol{\gamma}}^{(0)}\leftarrow\mathbf{W}_p\mathbf{r}$.
\STATE Reshape $\hat{\boldsymbol{\gamma}}^{(0)}$ into
$\hat{\boldsymbol{\Gamma}}^{(0)}$.
\STATE $\hat{\mathbf{G}}^{(0)}\leftarrow
\boldsymbol{\Phi}\hat{\boldsymbol{\Gamma}}^{(0)}$.
\FOR{$i=0,\ldots,I_{\max}-1$}
    \STATE $\hat{\mathbf{s}}_u^{(i)}\leftarrow
    \mathsf{Dec}(\mathbf{r};\hat{\mathbf{G}}^{(i)})$.
    \STATE Build $\boldsymbol{\Psi}_{u}(\hat{\mathbf{s}}_u^{(i)})$ and form
    $\boldsymbol{\Psi}(\hat{\mathbf{s}}_u^{(i)})
    =\boldsymbol{\Psi}_p+\boldsymbol{\Psi}_{u}(\hat{\mathbf{s}}_u^{(i)})$.
    \STATE Update $\hat{\boldsymbol{\gamma}}^{(i+1)}$ by
    \eqref{eq:decision-directed-ls-update}.
    \STATE Reshape $\hat{\boldsymbol{\gamma}}^{(i+1)}$ into
    $\hat{\boldsymbol{\Gamma}}^{(i+1)}$.
    \STATE $\hat{\mathbf{G}}^{(i+1)}\leftarrow
    \boldsymbol{\Phi}\hat{\boldsymbol{\Gamma}}^{(i+1)}$.
    \STATE $\epsilon_i^{\mathrm{upd}}\leftarrow
    (NL)^{-1}\|\hat{\mathbf{G}}^{(i+1)}-\hat{\mathbf{G}}^{(i)}\|_F^2$.
    \IF{$\epsilon_i^{\mathrm{upd}}\leq\epsilon_{\mathrm{stop}}$}
        \STATE \textbf{break}
    \ENDIF
\ENDFOR
\STATE \textbf{return} latest $\hat{\mathbf{s}}_u^{(i)}$ and
$\hat{\mathbf{G}}^{(i+1)}$.
\end{algorithmic}
\end{algorithm}

Initialization requires only the precomputed projection $\mathbf{W}_p\mathbf{r}$,
while each refinement costs $O(N(QL)^2+(QL)^3)$ plus the detector cost, dominated
by forming and solving the data-dependent normal matrix
$\boldsymbol{\Psi}^{H}(\hat{\mathbf{s}}_u^{(i)})\boldsymbol{\Psi}(\hat{\mathbf{s}}_u^{(i)})$.

\begin{figure*}[t]
\centering
\subfloat[BER]{
\includegraphics[width=0.45\textwidth]{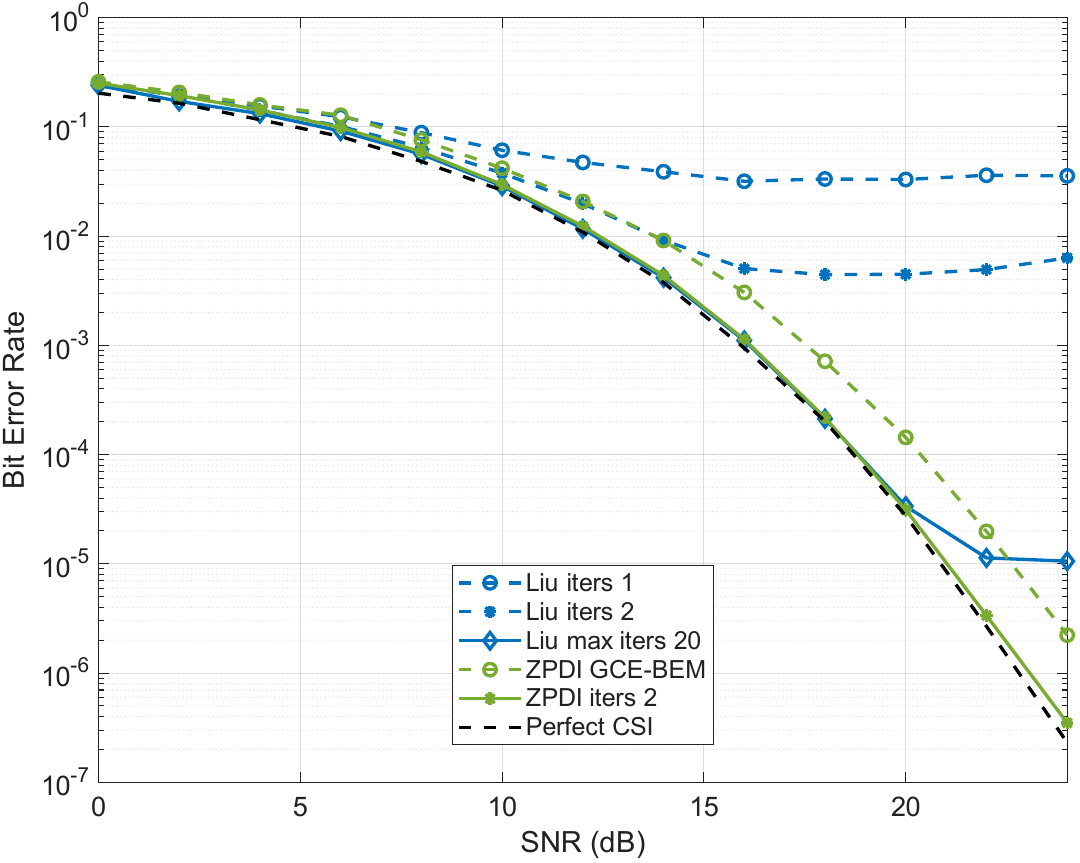}
\label{fig:dd-ber-v500}
}
\hfill
\subfloat[MSE]{
\includegraphics[width=0.45\textwidth]{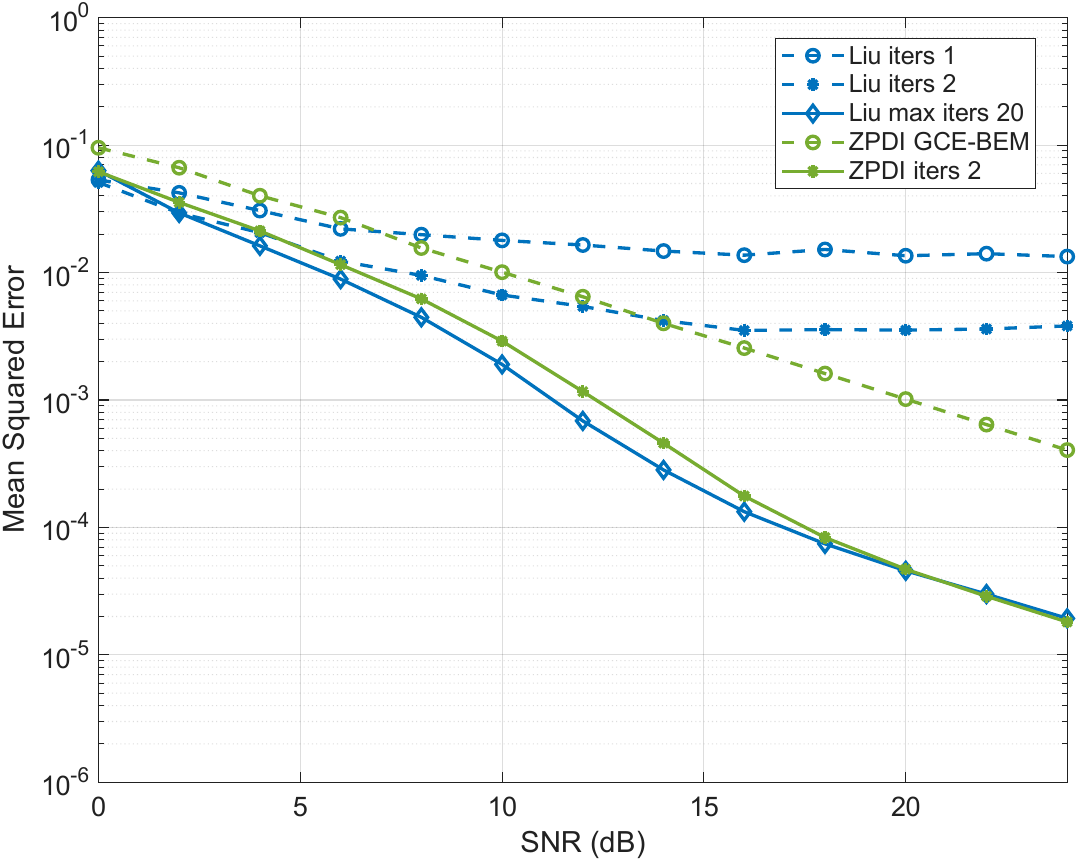}
\label{fig:dd-mse-v500}
}
\caption{Decision-directed receiver results at $v=500$~km/h: (a) BER and (b) channel-estimation MSE. The standalone ZPDI GCE-BEM curve uses $(R,Q)=(2,9)$. The ZPDI-initialized curves apply two data-aided LS refinements with the same basis. The Liu receiver~\cite{liu2022} uses $(R,Q_S)=(1,15)$ for initialization and $(R,Q_L)=(2,9)$ for refinement; its curves show one iteration, two iterations, and at most 20 iterations.}
\label{fig:dd-receiver-v500}
\end{figure*}

Fig.~\ref{fig:dd-receiver-v500} compares the two receivers. Two refinements
after ZPDI initialization bring the BER close to perfect CSI and below the
standalone pilot-only ZPDI curve. The Liu receiver retains high-SNR BER floors
after one and two iterations; a limit of 20 iterations reduces the floor but
stays above the ZPDI-initialized curve. The MSE behaves analogously: two
data-aided refinements after ZPDI initialization match the 20-iteration Liu
result at high SNR. Since both branches share the same refinement basis and LS
update, the difference is due to the pilot-only initialization.

\section{Conclusion}
\label{sec:conclusion}

This paper established exact conditions for payload-decoupled pilot-only BEM channel estimation. The result identifies when the usual matched-pilot LS estimate is payload-independent and equals the ML estimator for the reduced pilot statistic. If ZPDI fails, residual PDI becomes a deterministic, channel-scaled bias that produces a high-SNR error floor at any fixed pilot-to-data power ratio and cannot be removed by averaging. The operator-level condition covers arbitrary BEM subspaces, any unitary observation transform, and CP, ZP, and CPP guards. A basis-agnostic disjoint-support rule provides an explicit criterion for pilot, guard, and data placement. It gives an exact design test for a candidate pilot pattern. The separate pilot Gram matrix then governs identifiability, conditioning, and the cost of the single offline inversion.

Under ZPDI, the online estimator requires one precomputed projection with complexity linear in the frame length, while its error separates into thermal noise and BEM modeling error. For the OTFS single-pilot full-guard pattern, ZPDI and zero pilot--pilot leakage hold simultaneously. Controlled exact-BEM examples confirmed the predicted leakage, bias-plus-noise decomposition, high-SNR floor, and pilot-Gram structure. Over the 3GPP TDL-B channel, the ZPDI estimator with GCE-BEM approached perfect-CSI BER at $125$ and $500$~km/h and reliably initialized decision-directed refinement, whereas CE-BEM and the reference estimators exhibited high-SNR floors.

\bibliographystyle{IEEEtran}
\bibliography{IEEEabrv,OPTCE}

\end{document}